\documentclass[format=acmsmall, review=false, screen=true]{acmart-arXiv}
\usepackage{color}
\usepackage{graphicx}
\usepackage{xcolor}
\graphicspath{{./figures/}}
\DeclareGraphicsExtensions{.pdf}

\allowdisplaybreaks
\usepackage{booktabs}
\usepackage{enumitem}
\usepackage{hyperref}

\usepackage{prettyref}
\newrefformat{sec}{Section~\ref{#1}}
\newrefformat{subsec}{Section~\ref{#1}}
\newrefformat{subsubsec}{Section~\ref{#1}}
\newrefformat{thm}{Theorem~\ref{#1}}
\newrefformat{THM}{THEOREM~\ref{#1}}
\newrefformat{lemma}{Lemma~\ref{#1}}
\newrefformat{LEMMA}{LEMMA~\ref{#1}}
\newrefformat{prop}{Proposition~\ref{#1}}
\newrefformat{PROP}{PROPOSITION~\ref{#1}}
\newrefformat{corollary}{Corollary~\ref{#1}}
\newrefformat{fig}{Figure~\ref{#1}}
\newrefformat{tab}{Table~\ref{#1}}
\newrefformat{eq}{\eqref{#1}}
\newrefformat{app}{Appendix~\ref{#1}}
\newrefformat{ex}{Example~\ref{#1}}
\newrefformat{result}{Result~\ref{#1}}
\newrefformat{rmk}{Remark~\ref{#1}}

\usepackage[mathscr]{euscript}
\newcommand{\cG}{\mathscr{G}}
\renewcommand{\P}{\mathbb{P}}
\newcommand{\N}{\mathbb{N}}
\newcommand{\E}{\mathbb{E}}

\newcommand{\Z}{\mathbb{Z}}

\usepackage{bbm}
\newcommand{\ind}[1]{\mathbbm{1}_{#1}}

\newcommand{\var}{\textup{var}}
\newcommand{\LRU}{\textup{LRU}}
\newcommand{\TTL}{\textup{TTL}}

\newcommand{\F}{\mathscr{F}}

\makeatletter
\newtheorem*{rep@theorem}{\rep@title}
\newcommand{\newreptheorem}[2]{%
\newenvironment{rep#1}[1]{%
 \def\rep@title{#2 \ref{##1}}%
 \begin{rep@theorem}}%
 {\end{rep@theorem}}}
 
 \DeclareRobustCommand*\cal{\@fontswitch\relax\mathcal}
\makeatother
\newreptheorem{proposition}{Proposition}

\newtheorem{result}{Result}
\theoremstyle{remark}
\newtheorem{remark}{Remark}

\setcopyright{none}

\begin{document}
\title[On the Convergence of the TTL Approximation for an LRU Cache]{On the Convergence of the TTL Approximation for an LRU Cache under Independent Stationary Request Processes}
\author{Bo Jiang}
\affiliation{%
  \department{College of Information and Computer Sciences}
  \institution{University of Massachusetts Amherst}
  \city{Amherst}
  \state{MA}
  \postcode{01003}
  \country{USA}}
 \email{bjiang@cs.umass.edu}
\author{Philippe Nain}
\affiliation{%
  \institution{Inria}
  \streetaddress{Ecole Normale Sup\'erieure de Lyon, LIP, 46 all\'ee d'Italie}
  \city{69364 Lyon} 
  \country{France}}
 \email{philippe.nain@inria.fr}
\author{Don Towsley}
\affiliation{%
  \department{College of Information and Computer Sciences}
  \institution{University of Massachusetts Amherst}
  \city{Amherst}
  \state{MA}
  \postcode{01003}
  \country{USA}}
\email{towsley@cs.umass.edu}

\begin{abstract}

The modeling and analysis of an LRU cache is extremely challenging as exact results for the main performance metrics (e.g. hit rate) are either lacking or cannot be used because of their high computational complexity for large caches. As a result, various approximations have been proposed. 
The state-of-the-art method is the so-called TTL approximation,  first proposed and shown to be asymptotically exact for IRM requests by Fagin \cite{fagin1977asymptotic}. It has been applied to various other workload models and numerically demonstrated  to be accurate but without theoretical justification. 
 In this paper we provide theoretical justification for the approximation in the case where distinct contents are described by independent stationary and ergodic processes.  We show that this approximation is exact as the cache size and the number of contents go to infinity.  This extends earlier results for the independent reference model.  Moreover, we establish results not only for the aggregate cache hit probability but also for every individual content.  Last, we obtain bounds on the rate of convergence.

\end{abstract}

%
%


%
%



\thanks{This research was sponsored by the U.S. ARL and the U.K. MoD under
Agreement Number W911NF-16-3-0001 and by the NSF under Grants
CNS-1413998 and CNS-1617437. The views and conclusions contained in this document are those of the authors and should not be interpreted as representing the official policies, either expressed or implied, of the National Science Foundation, U.S. Army Research Laboratory, the U.S. Government, the U.K. Ministry of Defence or the U.K. Government. The U.S. and U.K. Governments are authorized to reproduce and distribute reprints for Government purposes notwithstanding any copyright notation hereon.}
%

\maketitle



\section{Introduction}
\label{sec:intro}

Caches are key components of many computer networks and systems. Moreover, they are becoming increasingly more important with the current development of new content-centric network architectures. A variety of cache replacement algorithms have been introduced and analyzed over the last few decades, mostly based on the least recently used algorithm (LRU).  Considerable work has focused on analyzing these policies \cite{bitner79,burville73,flajolet,king72,jelenkovic99,jelenkovic04,jelenkovic06}. 
Since exact results for the main performance metrics (e.g. hit rate) are either lacking or cannot be used because of their high computational complexity for large caches, approximations have been proposed \cite{fagin1977asymptotic,Dan90,jelenkovic1999asymptotic,che2002hierarchical,hirade2010analysis,osogami2010fluid, rosensweig}.
Of all the approximation techniques developed, the 
state of the art is provided by  the so-called TTL approximation based on time-to-live (TTL) caches, which has been demonstrated to be accurate for various caching policies and traffic models \cite{fagin1977asymptotic,che2002hierarchical,fricker2012versatile,bianchi2013check,fofackTTL,dehghan2015analysis,Leonardi16,leonardi2017modeling,gast2017ttl}. In this paper, we focus on the 
TTL approximation for the LRU cache with stationary requests.    
In a TTL cache, a time-to-live timer is set to its maximum value $T$ each time the content is requested. The content is evicted from the cache when the timer expires.

The link between an LRU cache and a TTL\footnote{Fagin worked with the so-called working-set policy, which is the discrete time version of the TTL policy. The result can be easily translated into one for the TTL approximation - also referred to as the Che's approximation in the literature, following the work of Che et al. in \cite{che2002hierarchical} -  under Poisson requests.} cache was first pointed out in \cite{fagin1977asymptotic} for  i.i.d. requests (the so-called independence reference model - IRM). In this paper, Fagin introduced the concept of a characteristic time (our terminology) and showed asymptotically that the performance of LRU converges to that of a TTL cache with a timer set to the characteristic time.   With the exception of an application to caching in \cite{flajolet}, this work {went unnoticed} and \cite{che2002hierarchical}  reintroduced the approximation, without theoretical justification, for LRU under Poisson requests. Fricker et al \cite{fricker2012versatile} provided some theoretical justification for the approximation by establishing a central limit theorem of the characteristic time under Poisson requests (see Remark \ref{remark:Fricker-et-al} in \prettyref{subsec:proof-convergence} for a brief discussion). More recently, \cite{Leonardi16} extended the TTL approximation to a setting where requests for distinct contents are independent and described by renewal processes.  The accuracy of this approximation is supported by simulations but a theoretical basis is lacking.  For independent Markovian Arrival Processes, \cite{gast2017ttl} developed TTL approximations for the more complicated LRU(m) and h-LRU policies, both including LRU as a special case. All the aforementioned  work focused on  stationary request processes with no dependence between different contents. Dependent and so-called time-asymptotically stationary requests were considered in \cite{osogami2010fluid}, but the results therein do not apply to the TTL approximation (see Section \ref{sec:convergence} for a brief discussion of this work).  Non-stationary request processes were considered in \cite{leonardi2017modeling}, where  a TTL approximation is developed for the hit probability in a single LRU cache and in a tandem of LRU caches, under the so-called shot noise request model. It is also shown  in \cite{leonardi2017modeling} that the cache eviction time converges to the characteristic time of the TTL approximation as the cache size goes to infinity.

The objective of the present paper is to provide a rigorous theoretical justification of the  TTL approximation for LRU in \cite{Leonardi16} and its generalization to independent stationary content request processes. To the best of our knowledge, such a justification was only provided in \cite{fagin1977asymptotic}, and later on in  \cite{hirade2010analysis}, under IRM (see Section \ref{ssec:ttl-approximation} for a discussion of Theorem 1 in \cite{hirade2010analysis}).
 
We make the following contributions in this paper.  First, we prove under the assumption that requests to distinct contents are described by mutually independent stationary and ergodic point processes, that the hit probability for each content under LRU converges to that for a TTL cache operating with a single timer value, called the LRU characteristic time, independent of the content.  Moreover, we derive rates of convergence for individual content hit probabilities under LRU to those under TTL using the LRU characteristic time. Under additional mild conditions, we then derive expressions for the characteristic time and the aggregate hit probability in the limit as the cache size and the number of contents go to infinity.  This last result extends the results of Fagin \cite{fagin1977asymptotic} for the independence reference model to a more general setting of independent stationary and ergodic content request processes. 

The rest of the paper is organized as follows.   Section \ref{sec:model} presents our model of an LRU cache under a general request model.  Section \ref{sec:overview} presents the main results of our paper. Section \ref{sec:convergence} proves the main result of the paper, namely the convergence of hit probabilities under LRU to those under TTL with bounds on the rate of convergence given in Section \ref{sec:rate}. Section \ref{sec:Fagin} extends Fagin's results to the more general case of stationary and ergodic request processes.  Last concluding statements are provided in Section \ref{sec:concl}.


\section{Model and Background}\label{sec:model}

We introduce the model for content request processes in \prettyref{subsec:request} and the content popularity in \prettyref{subsec:popularity}. \prettyref{subsec:CTA} presents the TTL approximation that approximates hit probabilities of an LRU  cache by those of a TTL cache with an appropriately chosen timer value.

\subsection{Content Request Process}\label{subsec:request}

We consider a cache of size $C_n$ serving $n$ unit sized contents labelled $i=1,\ldots , n$, where $C_n \in (0,n)$.
We assume that $C_n\to\infty$ as $n\to \infty$. 
In particular, several results will be obtained under the assumption that  $C_n \sim \beta_ 0 n$ with $\beta_0\in (0,1)$.
Requests for the contents are described by $n$ independent  stationary and ergodic simple point processes $N_{n,i}:=\{t_{n,i}(k),k\in \Z\}$, 
where $-\infty\leq \cdots < t_{n,i}(-1)< t_{n,i}(0)\leq 0< t_{n,i}(1)<\cdots\leq \infty$ represent successive request times to content $i=1,\ldots,n$. We assume the point processes are defined on a common probability space with probability measure $\P$ and associated expectation operator $\E$.
Let $0<\lambda_{n,i}<\infty$ denote the intensity of request process $N_{n,i}$, i.e., the long term average request rate for content $i$  (see  e.g. \cite[Sections 1.1 and 1.6]{BB-book-2003} for an introduction to stationary and ergodic point processes). Note that $\P[t_{n,i}(0)=0] = 0$ for all $i$ \cite[Section 1.1.4]{BB-book-2003}, i.e. no request arrives precisely at time 0. The same request processes were considered in \cite{ferragut2016optimizing} for TTL caches.

Following \cite{fofackTTL}, we will use Palm calculus for stationary and ergodic point processes \cite{BB-book-2003}. Let $\P^0_{n,i}$  be the Palm probability\footnote{Readers unfamiliar with Palm probability can think of $\P^0_{n,i}$ as being defined by $\P^0_{n,i}[A]=\P[A\mid t_{n,i}(0)=0]$ for any event $A$, i.e.~the conditional probability conditioned on the event that content $i$ is requested at time $0$, although the definition is more general.} associated with the point process $N_{n,i}$ (see e.g. \cite[Eq.~(1.2.1)]{BB-book-2003}). 
In particular, $\P^0_{n,i}[t_{n,i}(0)=0]=1$, i.e.~under $\P^0_{n,i}$ content $i$ is requested at time $t=0$.  It is known
that  \cite[Exercice 1.2.1]{BB-book-2003}
\begin{equation}
\E^0_{n,i}[t_{n,i}(1)]=\frac{1}{\lambda_{n,i}},
\label{mean}
\end{equation}
where $\E^0_{n,i}$ is the expectation operator associated with $\P^0_{n,i}$.  Define 
\begin{equation}\label{eq:G_i}
G_{n,i}(t)=\P^0_{n,i}[t_{n,i}(1)\leq  t],
\end{equation}
the cdf of the inter-request time for content $i$ under $\P^0_{n,i}$.

For any distribution $F$, we denote its mean by $m_F$ and the corresponding ccdf by $\bar F :=1-F$. For any $F$ with support in $[0,\infty)$ and $m_F\in (0,\infty)$, we define an associated  distribution $\hat F$ by
\begin{equation}\label{eq:F-hat}
\hat F(t) = \frac{1}{m_F} \int_0^t \bar F(z) dz, \quad t\geq 0.
\end{equation}
It is  well-known that (see e.g. \cite[Section 1.3.4]{BB-book-2003}) 
\begin{equation}\label{eq:age}
\P[-t_{n,i}(0) \leq t] = \hat G_{n,i}(t) = \lambda_{n,i}  \int_0^t \bar G_{n,i}(z)dz,
\end{equation}
with  $m_{G_{n,i}} = 1/\lambda_{n,i}$ from \eqref{mean}. Note that  $\P[-t_{n,i}(0) \leq t]$ is the cdf  of the time elapsed since content $i$ was last requested before the random observation time $t=0$ (recall that the system is in steady state at time $t=0$), often referred to as the age distribution of the last request for content $i$.

We assume all cdfs $G_{n,i}$ are continuous.  Let 
\begin{equation}\label{eq:G_ast}
 G_{n,i}^{\ast}(t) = G_{n,i}(t/\lambda_{n,i})
\end{equation}
be the scaled version of $G_{n,i}$ that is standardized in the sense that it has  unit mean. 
We assume that there exists a continuous cdf $\Psi$ with support in $[0,\infty)$ and mean $m_{\Psi}>0$  such that 
\begin{equation}
\label{eq:assumption-Psi}
\bar G_{n,i}^{\ast}(t) \geq \bar \Psi(t), \quad \forall t,n,i,
\end{equation}
or, by the definition of $G_{n,i}^\ast$,
\begin{equation}\label{eq:Psi-bar}
\bar G_{n,i}(t) \geq \bar \Psi(\lambda_{n,i} t), \quad \forall t, n,i,
\end{equation}
which, by \prettyref{eq:F-hat}, implies 
\begin{equation}\label{eq:Psi-hat}
 \hat G_{n,i}(t)\geq m_{\Psi} \hat \Psi(\lambda_{n,i} t), \quad \forall t, n, i.
\end{equation}

Let us elaborate a bit on the assumption in \eqref{eq:assumption-Psi}. Consider the $L_1$ distance between $\Psi$ and $G_{n,i}^{\ast}$, which, by \eqref{eq:assumption-Psi}, is given by
\[
\|G_{n,i}^{\ast} - \Psi\|_1=\|\bar G_{n,i}^{\ast} - \bar \Psi\|_1 = \int_0^\infty [\bar G_{n,i}^{\ast}(t) - \bar \Psi(t)]dt = 1 - m_\Psi.
\]
Since $\|G_{n,i}^{\ast} - \Psi\|_1\geq 0$, it follows that $m_\Psi\leq 1$.  Note that all $G_{n,i}^\ast$ live on the sphere of radius $1-m_\Psi$ centered at $\Psi$.  Since both $G_{n,i}^\ast$ and $\Psi$ are continuous, $m_\Psi=1$ if and only if $G_{n,i}^{\ast}(t)=\Psi(t)$ or, equivalently, if and only if
$G_{n,i}(t)=\Psi (\lambda_{n,i} t)$ for all $t$, $n$ and $i$. Intuitively, the function $\Psi$ controls the variability within the family of cdfs $\cG=\{G_{n,i}^\ast: n\geq i\geq 1\}$, and $m_\Psi$ is a measure of this variability. When $m_\Psi\to 0$, the constraint \eqref{eq:assumption-Psi} becomes empty, and $G_{n,i}^\ast$ could be very different from each other. As $m_\Psi$ increases, $G_{n,i}^\ast$ become more and more similar to each other. When $m_\Psi=1$, $G_{n,i}^\ast$ degenerates to a single distribution $\Psi$, in which case, $G_{n,i}$ are all from the scale family\footnote{Recall that a family of cdfs $F(st)$, indexed by a \emph{scale parameter} $s>0$, is called the scale family with standard cdf $F$.} as $G_{n,i}(t)=\Psi(\lambda_{n,i} t)$ from (\ref{eq:G_ast}). 

The most important example of the degenerate case $m_\Psi=1$ is when all request processes are Poisson, i.e.~$G_{n,i}(t) = 1-e^{-\lambda_{n,i}t}$ with $\Psi(t) = 1-e^{-t}$. Non-Poisson examples include Erlang distributions with the same number of stages, Gamma distributions with the same shape parameter, and Weibull distributions with the same shape parameter.

An important example of the non-degenerate case is when $G_{n,i}$ are from a finite number, $J$, of scale families, i.e.~$\cG=\{\Psi_1,\dots,\Psi_J\}$ for some distinct cdfs $\Psi_j$ with $m_{\Psi_j}=1$. More specifically, let ${\cal P}_1,\ldots,{\cal P}_J$ be a partition of the set $\{(n,i)\in \N^2:n\geq i\geq 1\}$ such that $G^\ast_{n,i}=\Psi_j$ for  all $(n,i)\in {\cal P}_j$.
Note that  \eqref{eq:assumption-Psi} holds with $\Psi(t)=\max_{1\leq j\leq J} \Psi_j(t)$ in this case. However, $m_\Psi<1$ unless $J=1$, which reduces to the degenerate case.

\vspace{.5em}
Let $N_n:=\{t_n(k),k\in\Z\}$ be the point process resulting from the superposition of the $n$ independent point processes $N_{n,1},\ldots,N_{n,n}$, where $-\infty\leq \cdots < t_n(-1)< t_n(0)\leq 0< t_n(1)<\cdots\leq \infty$. Note that we have used the fact that the points $t_n(k)$ are distinct with probability one \cite[Property 1.1.1]{BB-book-2003}.
Let $\P_n^0$ be the Palm probability\footnote{Again,  readers unfamiliar with  Palm probability can think of $\P_n^0$ as being defined by $\P_n^0[A]=\P[A\mid t_n(0)=0]$ for any event $A$, i.e.~the conditional probability conditioned on the event that a request arrives at time $0$.} associated with $N_n$, and $\E_n^0$ the associated expectation operator. Under $\P_n^0$ a content is requested at $t=0$, i.e.~$\P_n^0[t_n(0)=0]=1$. Let $X_{n}^0\in\{1,\ldots,n\}$
denote this content. It is known that (see e.g. \cite[Section 1.4.2]{BB-book-2003})
\begin{equation}
 \label{eq:p_i}
 \P_n^0[X_{n}^0=i]=\frac{\lambda_{n,i}}{\Lambda_n}:=p_{n,i},
\end{equation}
where $\Lambda_n:=\sum_{i=1}^n \lambda_{n,i}$, and
\begin{equation}\label{eq:total-prob}
\P_n^0[A] = \sum_{i=1}^n p_{n,i} \P_{n,i}^0[A]
\end{equation}
for any event $A$. 

\subsection{Content Popularity}\label{subsec:popularity}

The probability $p_{n,i}$ defined in \eqref{eq:p_i} gives the popularity of content $i$. Previous work (see e.g.  \cite{fricker2012versatile} and references therein)  shows  that the popularity distribution $\{p_{n,1},\dots,p_{n,n}\}$ usually follows Zipf's law, 
\begin{equation}\label{eq:Zipf}
p_{n,i} = \frac{i^{-\alpha}}{\sum_{j=1}^n j^{-\alpha}},
\end{equation}
where $\alpha \geq 0$ and most often $\alpha\in (0,1)$. This will be the main example of popularity distribution used throughout the rest of the paper.

In \cite{fagin1977asymptotic}, the popularity distribution
is assumed to be given by
\begin{equation}\label{eq:Fagin-pmf}
p_{n,i} = F\left(\frac{i}{n}\right) - F\left(\frac{i-1}{n}\right),
\end{equation}
where $F$ is a continuously differentiable cdf with support in $[0,1]$. With some slight modification, \eqref{eq:Fagin-pmf} can be extended to include \eqref{eq:Zipf} as a special case. Note that \eqref{eq:Fagin-pmf} does not assume the $p_{n,i}$'s are ordered in $i$.

In this paper, we consider more general popularity distributions, which include as special cases both \eqref{eq:Zipf} and \eqref{eq:Fagin-pmf} with the mild condition that $F' > 0$ a.e.~on $[0,1]$. Let $\sigma_i$ be the index of the $i$-th most popular content, i.e.
\begin{equation}\label{eq:p-decreasing}
p_{n,\sigma_1}\geq p_{n,\sigma_2}\geq  \dots   \geq p_{n,\sigma_n}
\end{equation}
is the sequence $p_{n,1}, \dots, p_{n,n}$ rearranged in decreasing order.  
Define the tail $\bar P_n$ of the content popularity distribution by
\begin{equation}\label{eq:P-bar}
\bar P_n(i) = \sum_{k=i+1}^n p_{n,\sigma_k},
\end{equation}
which is the aggregate popularity of the $n-i$ least popular contents. 
Roughly speaking, we will focus on popularity distributions whose values $\bar P_n(i)$ are of the same order for $i$ around $C_n$. This will be made more precise later; see assumption (P1) in \prettyref{subsubsec:popularity}. 

\subsection{TTL Approximation}\label{subsec:CTA}
\label{ssec:ttl-approximation}

Let $Y_{n,i}(t)=1$ if content $i$ is requested during the interval $[-t,0)$ and $Y_{n,i}(t) = 0$ otherwise.  With this notation, 
\begin{equation}\label{eq:Y}
Y_n(t) := \sum_{i=1}^n Y_{n,i}(t)
\end{equation}
is the number of distinct contents requested during $[-t,0)$.
Let $[-\tau_n,0)$ be the smallest past interval in which $C_n$ distinct contents are referenced, i.e., 
\begin{equation}\label{eq:tau}
\tau_n = \inf\{t: Y_n(t) \geq C_n\}.
\end{equation}
Note that if we reverse the arrow of time, we obtain statistically the same request processes, and $\tau_n$ is a stopping time for the process $Y_n(t)$.  

In an LRU cache, a content that is least recently referenced is evicted when another content needs to be added to the full cache. Thus a request for content $i$ results in a cache hit if and only if $i$ is among the $C_n$ distinct most recently referenced contents. By stationarity, we can always assume that this request arrives at $t=0$. Thus
the stationary hit probability of an LRU cache is given by
\begin{equation}\label{eq:H-LRU}
H_n^{\LRU}= \P_n^0[Y_{n,X_n^0}(\tau_n) = 1]. 
\end{equation}
Similarly,  the stationary hit probability of content $i$ in an LRU cache is  given by
\begin{equation}\label{eq:H_i-LRU}
H^{\LRU}_{n,i} = \P^0_{n,i}[Y_{n,i}(\tau_n) = 1],
\end{equation}
By \eqref{eq:total-prob}, $H_n^\LRU$ and $H_{n,i}^\LRU$ are related by
\begin{equation}\label{eq:H-H_i-LRU}
H_n^{\LRU}=\sum_{i=1}^n p_{n,i}  H^{\LRU}_{n,i}.
\end{equation}

In a TTL cache, when a content is added to the cache, its associated time-to-live timer is set to its maximum value $T$. The content is evicted from the cache when the timer expires. The capacity of the cache is assumed to be large enough to hold all contents with non-expired timers. In this paper, we consider the so-called TTL cache with reset, which always resets the associated timer to $T$ when a cache hit occurs. Thus a request for content $i$ results in a cache hit if and only if $i$ is referenced in a past window of length $T$. The stationary hit probability is then given by
\begin{equation}\label{eq:H-TTL}
H_n^{\TTL}(T) = \P_n^0[Y_{n,X_n^0}(T) = 1],
\end{equation}
and that for content $i$ by
\begin{equation}\label{eq:H_i-TTL}
H^{\TTL}_{n,i}(T) = \P^0_{n,i}[Y_{n,i}(T) = 1],
\end{equation}
which will be shown to equal $G_{n,i}(T)$ in \prettyref{lem:Y}.
By \eqref{eq:total-prob}, $H_n^\TTL(T)$ and $H^\TTL_{n,i}(T)$ are related by
\begin{equation} \label{eq:H-H_i-TTL}
H_n^{\TTL}(T) =\sum_{i=1}^n p_{n,i} H^{\TTL}_{n,i}(T).
\end{equation}

The TTL approximation was  first introduced by Fagin for IRM requests \cite{fagin1977asymptotic}, later rediscovered for independent Poisson request processes  \cite{che2002hierarchical} and extended to renewal request processes \cite{Leonardi16}, in the latter two cases without theoretical basis.  It should be noticed that Fagin's result can be reproduced \cite{Jelenkovic-private} by restricting the support of the distribution to $[0,1]$ in Theorem 4 in \cite{jelenkovic1999asymptotic}. 
Also, Theorem 1 in \cite{hirade2010analysis} proves that the individual content hit probability in an LRU cache converges to the corresponding quantity in a TTL cache as the number  of items increases to infinity, when contents are requested according to independent Poisson processes and when there is only a finite number of types of contents; see discussion after Example \ref{ex:heavy-tail-Zipf}. 

We now present it for general independent stationary and  ergodic request processes. Let 
\begin{equation}\label{eq:K}
K_n(T):=\E[Y_n(T)]
\end{equation}
denote the expected number of contents in a TTL cache with  timer value $T$, where $Y_n$ is defined in \eqref{eq:Y}. It will be shown in \prettyref{lem:K} that $K_n(T) = \sum_{i=1}^n \hat G_{n,i}(T)$. Given the size $C_n$ of an LRU cache, let $T_n$ satisfy 
\begin{equation}\label{eq:CT}
C_n = K_n(T_n) = \sum_{i=1}^n \hat G_{n,i}(T_n).
\end{equation}
The time $T_n$ is the \emph{characteristic time} of the LRU cache. The TTL approximation then approximates the hit probabilities of the LRU cache by those of a TTL cache with timer value $T_n$, i.e.
\[
H_{n,i}^\LRU \approx H_{n,i}^\TTL(T_n), \quad \forall i=1,\dots,n.
\]
For Poisson requests, \eqref{eq:CT} takes the familiar form
\[
C_n = \sum_{i=1}^n (1-e^{-\lambda_{n,i} T_n}).
\]
Note that the TTL approximation for general independent stationary and ergodic processes takes the same form as for renewal processes \cite{Leonardi16}, which is not surprising in view of Theorem 2 in \cite{gast2017ttl}.

In \prettyref{sec:convergence}, we show that, as $C_n$ and $n$ become large, the TTL approximation becomes exact, i.e. an LRU cache behaves like a TTL cache with a TTL approximation timer value equal to the LRU characteristic time.


\section{Overview of Main Results}\label{sec:overview}

In this section we present the main results of the paper. \prettyref{subsec:assumptions} collects various assumptions used in the main results and discusses their relations. The mains results are presented in \prettyref{subsec:main-results}.

\subsection{Assumptions}\label{subsec:assumptions}

We divide the assumptions into three categories according to whether they concern  cache size,  request processes, or content popularity distribution.

\subsubsection{Cache size}

Throughout the paper, it is assumed that the cache size $C_n\in (0,n)$ and $C_n \to \infty$ as $n\to\infty$. In addition, each result assumes one of the following conditions.

\begin{enumerate}[label=(C\arabic*), ref=C\arabic*]
\item\label{ass:C1} $C_n \leq \beta_1 n$ for some $\beta_1 \in (0,m_\Psi)$  and $n$ large enough, where $m_\Psi$ is the mean of $\Psi$ in \eqref{eq:assumption-Psi}.
\item\label{ass:C2} $C_n \sim \beta_0 n$ for some $\beta_0 \in (0,1)$.
\end{enumerate}

\vspace{.2em}
Note that \eqref{ass:C2} requires $C_n$ to scale linearly in $n$ while \eqref{ass:C1} only requires $C_n$ to scale at most linearly. For $\beta_0 < m_\Psi$, \eqref{ass:C2}$\implies$\eqref{ass:C1}.

\subsubsection{Request processes}\label{subsubsec:R-conditions}
The requests for different contents follow independent stationary and ergodic simple point processes. The request process for content $i$ has continuous inter-request distribution satisfying \eqref{eq:assumption-Psi}. 
In addition, each result assumes one of the following conditions, with  $\cG_i:= \{G^{\ast}_{n,i}\,:\, n\geq i\}$ and $\cG :=\bigcup_{i=1}^\infty \cG_i$, where $G^\ast_{n,i}$ is defined in \eqref{eq:G_ast},
\begin{enumerate}[label=(R\arabic*),ref=R\arabic*]
\item\label{ass:R1} Given $i$, $\cG_i$ is equicontinuous\footnote{A family of functions $\F$ is  \emph{equicontinuous} if for every $\epsilon>0$, there exists a $\delta>0$ such that $|x_1-x_2|<\delta$ implies $|f(x_1)-f(x_2)|<\epsilon$ for every $f\in \F$. There is another commonly used definition of equicontinuity, which is a weaker notion in general but turns out to be equivalent to the former in our setting.}.
\item\label{ass:R2}  $\cG$ is equicontinuous.
\item\label{ass:R3} $|\cG|<\infty$, i.e. the inter-request distributions are from a finite number of scale families.
\item\label{ass:R4} $\cG=\{\Psi\}$, i.e. the inter-request distributions are from a single scale family.
\item\label{ass:R5}$\cG$ is uniformly Lipschitz continuous\footnote{A family of functions $\F$ is \emph{uniformly Lipschitz continuous} if there exists an $M>0$ such that $|f(x_1)-f(x_2)|<M |x_1-x_2|$ for every $x_1, x_2$ and every $f\in \F$.}.
\item\label{ass:R6} There exist a constant $B$ and 
 $\rho \in (0,1]$  such that
 \begin{equation}\label{eq:Lipschitz}
 \left|G(t) - G(t\pm x t)\right| \leq B x, \quad \text{for\; $x \in [0,\rho]$, $\forall t$ and $\forall G\in \cG$}.
\end{equation}
\end{enumerate}

\vspace{.2em}
By \prettyref{lem:equicontinuity},  \eqref{ass:R1} (resp. \eqref{ass:R2}) holds if $\cG_i$ (resp. $\cG$) is composed of a finite family of continuous cdfs. Hence,
 \eqref{ass:R4}$\implies$\eqref{ass:R3}$\implies$\eqref{ass:R2}$\implies$\eqref{ass:R1}. Note also that  \eqref{ass:R5}$\implies$\eqref{ass:R2}. Examples of \eqref{ass:R5} include families of distributions that have densities with a common upper bound. The last condition \eqref{ass:R6} can be thought of as some kind of uniform Lipschitz continuity, where the bound depends on the relative deviation of the arguments rather than on the absolute deviation as in \eqref{ass:R5}. 
 Condition \eqref{ass:R6} is satisfied if the inter-request distributions are all exponential, which corresponding to Poisson requests (\prettyref{ex:Poisson}), 
or, more generally,  if \eqref{ass:R3} holds with every $G\in \cG$ having a continuous density (see \prettyref{ex:general}, which also includes an example with infinite $\cG$).
Note that \eqref{ass:R6} implies uniform Lipschitz continuity for $t$ strictly bounded away from zero, which is in fact all we need when working with \eqref{ass:R5}, so  for our purpose \eqref{ass:R6} is stronger than \eqref{ass:R5}. 

\subsubsection{Popularity distribution}\label{subsubsec:popularity}
Each result assumes one of the following conditions for content popularity distribution.

\begin{enumerate}[label=(P\arabic*),ref=P\arabic*]
\item\label{ass:P1} There exist constants  $\kappa_1 \in (\frac{1}{m_\Psi}, \frac{1}{\beta_1})$ for $\beta_1$ in \eqref{ass:C1}, $\kappa_2\in [0,1]$ and $\gamma\in (0,1)$ such that for all sufficiently large $n$, the tail popularity $\bar P_n$ defined in \eqref{eq:P-bar} satisfies
\begin{equation}\label{eq:heavy-tail}
\bar P_n(\lceil \kappa_1 C_n\rceil) > \gamma \bar P_n(\lfloor \kappa_2 C_n\rfloor).
\end{equation}

\item\label{ass:P2} Fagin's condition: for some continuous function $f$ defined on $(0,1]$ such that $f>0$ a.e.~and $\lim_{x\to 0+} f(x)\in [0,+\infty]$, and for some $z_{n,i}\in [\frac{i-1}{n},\frac{i}{n}]$,  the popularities $p_{n,i} \sim g_nf(z_{n,i})$ uniformly in $i$, i.e.
\begin{equation}\label{eq:p_i-gf}
\max_{1\leq i\leq n} \left|\frac{g_nf(z_{n,i})}{p_{n,i}} - 1\right|\to 0, \quad \text{as } n\to\infty.
\end{equation}

\item\label{ass:P3} The generalization \eqref{eq:p_i-gf-multi-F} of \eqref{ass:P2} from a single function $f$ to a finite number of functions $f_j$'s.
\end{enumerate}

\vspace{.3em}
For a discussion of \eqref{ass:P1}, see \prettyref{rmk:mu} after \prettyref{prop:convergence}. Note that \eqref{ass:P2} is slightly more general than Fagin's original condition \eqref{eq:Fagin-pmf}. Note also \eqref{ass:P2}$\implies$\eqref{ass:P3}$\implies$ \eqref{ass:P1}. \prettyref{ex:heavy-tail-Zipf}  shows that the Zipfian popularity distribution in \eqref{eq:Zipf} with $\alpha\geq 0$ satisfies \eqref{ass:P1}. \prettyref{ex:Fagin-Zipf} shows that it also satisfies \eqref{ass:P2} and hence \eqref{ass:P3}. 

\vspace{.5em}
The common assumptions that $C_n\to\infty$ as $n\to\infty$ and that requests for different contents are described by mutually independent stationary and ergodic processes satisfying \prettyref{eq:assumption-Psi} will be assumed without explicit mentioning throughout the rest of the paper.

\subsection{Main Results}\label{subsec:main-results}

In this section we present the main results of the paper. 
The first establishes that individual content hit probabilities under LRU converge to those under TTL as the cache size $C_n$ and the number of contents $n$ go to infinity, provided the timer values for all contents are set to the LRU characteristic time $T_n$ introduced in the previous section, and  provided the inter-request time distributions satisfy certain continuity properties.

\begin{result}[\prettyref{prop:convergence}]\label{result:convergence}
Under assumptions \eqref{ass:C1}, \eqref{ass:R1} and \eqref{ass:P1}, TTL approximation is asymptotically exact for content $i$, i.e.
\[
\left|H^\LRU_{n,i} - H^\TTL_{n,i}(T_n)\right| \to 0, \quad\hbox{as } n\to\infty.
\]
Under assumptions \eqref{ass:C1}, \eqref{ass:R2} and \eqref{ass:P1}, TTL approximation is asymptotically exact uniformly for all contents, i.e.
\[
\max_{1\leq i\leq n}\left|H^\LRU_{n,i} - H^\TTL_{n,i}(T_n)\right| \to 0, \quad\hbox{as } n\to\infty.
\]
\end{result}

The next result provides a uniform bound for the rates at which individual content hit probabilities under LRU converge to those under TTL under a slightly stronger Lipschitz continuity property.  

\begin{result}[\prettyref{prop:rate}]\label{result:rate}
Under assumptions \eqref{ass:C1}, \eqref{ass:R5} and \eqref{ass:P1}, the following holds,
\[
\max_{1\leq i\leq n}\left|H^\LRU_{n,i} - H^\TTL_{n,i}(T_n)\right|  = O\left(\left(\frac{\log C_n}{C_n}\right)^{\frac{1}{4}}\right).
\]
\end{result}

The above rate of convergence is slow.  This is improved in the next result where it is shown to be $O((\log C_n/C_n)^{1/2})$ under slightly stronger assumptions regarding the marginal inter-request time distributions. However, numerical results (see e.g \cite{fricker2012versatile}) suggest that the convergence rate might be faster than proved here.

\begin{result}[\prettyref{prop:rate-faster}]\label{result:rate-faster}
Under assumptions \eqref{ass:C1}, \eqref{ass:R6} and \eqref{ass:P1}, the following holds,
\[
\max_{1\leq i\leq n}\left|H^\LRU_{n,i} - H^\TTL_{n,i}(T_n)\right|  = O\left(\sqrt{\frac{\log C_n}{C_n}} \right).
\]
\end{result}

The last two results include extensions of Fagin's results for IRM 
to the case where content requests are described by mutually independent stationary and ergodic processes where the marginal inter-request time distributions satisfy mild continuity properties.

\begin{result}[\prettyref{prop:H-limit}]\label{result:H-limit}
Under assumptions \eqref{ass:C2}, \eqref{ass:R4} and \eqref{ass:P2}, the following holds,
\[
H_n^\LRU \to  \int_0^1 f(x) \Psi(\nu_0 f(x)) dx,  \quad\hbox{as } n\to\infty,
\]
where $\nu_0$ the unique real number in $(0,\infty)$ that satisfies
\[
\int_0^1 \hat \Psi(\nu_0 f(x))dx = \beta_0.
\]
\end{result}

\prettyref{result:H-limit} considers a single class of contents in the sense that there is a single $f$ and a single $\Psi$ for all contents. The following result extends it to $J$ classes of contents, where class $j$ has a fraction $b_j$ of the total contents, and each class $j$ satisfies the assumptions in \prettyref{result:H-limit} with potentially different $f_j$ and $\Psi_j$.  See \prettyref{prop:H-limit-J} for a more precise statement of   \eqref{ass:R3} and \eqref{ass:P3}.

\begin{result}[\prettyref{prop:H-limit-J}]\label{result:H-limit-multi-F} Under assumptions \eqref{ass:C2}, \eqref{ass:R3} and \eqref{ass:P3}, the following holds,
\[
 H^{\LRU}_{n} \to \sum_{j=1}^J  b_j\int_0^1  f_j(x)\Psi_j(\nu_0 f_j(x)) dx,  \quad\hbox{as } n\to\infty,
\]
where $\nu_0$ the unique real number in $(0,\infty)$ that satisfies 
\[
 \sum_{j=1}^J b_j\int_0^1  \hat \Psi_j(\nu_0 f_j(x))dx=\beta_0.
\]
\end{result}


\section{Asymptotic Exactness}
\label{sec:convergence}

It has been observed numerically in \cite{fricker2012versatile} that the TTL approximation is very accurate uniformly for contents of a wide range of popularity rank when the request processes are all Poisson. In this section, we prove that under some general conditions, the TTL approximation is exact in the large system regime, in the sense that individual content hit probabilities under LRU converge uniformly to those under TTL using the LRU characteristic time.

The following bounds on the LRU characteristic time $T_n$, which may be of interest in their own right,
will be used in the proof of the main result, \prettyref{prop:convergence}. The proof is found in \prettyref{subsec:proof-T_n-bounds}.

\begin{proposition}\label{prop:T_n-bounds}
The characteristic time $T_n$ defined by \eqref{eq:CT} exists and is unique. 
For any $n_1 \in (C_n/m_\Psi, n]$, which exists if $C_n < n m_\Psi$, we have
\begin{equation}\label{eq:T_n-upper}
T_n \leq \frac{\nu_0}{\lambda_{n,\sigma_{n_1}}},
\end{equation}
where $\sigma_{n_1}$ is defined in \eqref{eq:p-decreasing}, and $\nu_0$, which exists, is any constant that satisfies 
\[
\hat \Psi(\nu_0) \geq \frac{C_n}{n_1 m_\Psi}.
\]
For any $n_2 \leq C_n$,
\begin{equation}\label{eq:T_n-lower}
T_n \geq \frac{C_n-n_2}{\Lambda_n \bar P_n(n_2)}.
\end{equation}
\end{proposition}

The following examples show that \prettyref{prop:T_n-bounds} yields the same scaling order of $T_n$ as in \cite[Eq.~(7)]{fricker2012versatile} for Zipfian popularity distribution with $\alpha\neq 1$, but for request processes more general than Poisson.

\begin{example}
Consider Zipfian popularity distribution in \eqref{eq:Zipf} with $\alpha \in (0,1)$. In this case, we need $C_n = \Omega(n)$ so that the cache stores a nonnegligible fraction of the files in the sense that $P_n(C_n)$ does not vanish as $n$ increases. Assume $C_n\sim \beta_0 n$ with $\beta_0 \in (0, m_\Psi)$. Setting $n_1 = n$ in \eqref{eq:T_n-upper}, we obtain
\[
T_n \leq \frac{\nu_0}{p_{n,n}\Lambda_n } \sim \frac{\nu_0  n}{(1-\alpha)\Lambda_n},
\]
where $\nu_0$ satisfies $\hat\Psi(\nu_0) > \beta_0/m_\Psi$. Setting $n_2=0$ in \eqref{eq:T_n-lower}, we obtain
\[
T_n \geq \frac{C_n}{\Lambda_n} \sim \frac{\beta_0 n}{\Lambda_n}.
\]
Note that
\[
p_{n,\sigma_n} = p_{n,n} = \frac{n^{-\alpha}}{\sum_{j=1}^n j^{-\alpha}} \sim (1-\alpha)n^{-1}, 
\]
where the last step follows from the well-known asymptotics (see e.g. \cite[Theorem 3.2]{apostol1976})
$\sum_{j=1}^n j^{-\alpha} \sim  n^{1-\alpha}/(1-\alpha)$ for large $n$.
Therefore,
$T_n = \Theta\left(n\Lambda_n^{-1}\right)$.
In particular, if $\lambda_{n,i} = i^{-\alpha}$, then $\Lambda_n\sim  n^{1-\alpha}/(1-\alpha)$ and hence $T_n = \Theta(n^{\alpha})$. 
\end{example}

\begin{example}
Consider Zipfian popularity distribution in \eqref{eq:Zipf} with $\alpha > 1$. In this case, $P_n(C_n)$ never vanishes as long as $C_n\geq 1$. Assume $C_n\leq \beta_0 n$ with $\beta_0 \in (0, m_\Psi)$. Consider the limit $C_n\to\infty$. Setting $n_1 \sim \kappa_1 C_n$ in \eqref{eq:T_n-upper} with $\kappa_1 \in (\frac{1}{m_\Psi},\frac{1}{\beta_0})$, we obtain
\[
p_{n,\sigma_{n_1}} = p_{n,n_1} = \frac{n_1^{-\alpha}}{\sum_{j=1}^n j^{-\alpha}} \sim \frac{1}{\kappa_2^{\alpha} C_n^{\alpha}\zeta(\alpha)}, 
\]
and hence
\[
T_n \leq \frac{\nu_0  \kappa_1^{\alpha}\zeta(\alpha) C_n^\alpha}{\Lambda_n},
\]
where $\nu_0$ satisfies $\hat\Psi(\nu_0) > (\kappa_1 m_\Psi)^{-1}$.  Setting $n_2\sim \kappa_2 C_n$ in \eqref{eq:T_n-lower} with $\kappa_2 \in (0,1)$, we obtain
\[
\bar P_n(n_2) \sim \frac{n_2 ^{1-\alpha}}{(1-\alpha)\zeta(\alpha)},
\]
where $\zeta(\alpha) = \sum_{j=1}^\infty j^{-\alpha}$ is the Riemann zeta function. Thus
\[
T_n \geq \frac{C_n-n_2}{\Lambda_n\bar P(n_2)} \sim \frac{(1-\alpha)\zeta(\alpha)(C_n-n_2)}{\Lambda_n  n_2^{1-\alpha} } \sim \frac{(1-\alpha)\zeta(\alpha)(1-\kappa_2)C_n^\alpha}{\Lambda_n  \kappa_2^{1-\alpha} },
\]
Therefore, 
$T_n = \Theta\left(C_n^\alpha \Lambda_n^{-1}\right)$.
In particular, if $C_n=\Theta(n)$ and $\lambda_i = i^{-\alpha}$, then $\Lambda_n \sim \zeta(\alpha)$ and hence $T_n = \Theta(n^\alpha)$. However, we do not need to have $C_n$ scale linearly in $n$.
\end{example}

\prettyref{prop:convergence} is the main result, which provides sufficient conditions for the hit probabilities  in the TTL approximation to converge to the corresponding hit probabilities in the LRU cache. The proof is found in \prettyref{subsec:proof-convergence}.

\begin{proposition}
\label{prop:convergence}
Under assumptions \eqref{ass:C1}, \eqref{ass:R1} and \eqref{ass:P1}, TTL approximation is asymptotically exact for content $i$, i.e.
\begin{equation}
\label{eq:convergence}
\left|H^\LRU_{n,i} - H^\TTL_{n,i}(T_n)\right| \to 0, \quad\hbox{as }\ n\to\infty.
\end{equation}
Under assumptions \eqref{ass:C1}, \eqref{ass:R2} and \eqref{ass:P1}, TTL approximation is asymptotically exact uniformly for all contents, i.e.
\begin{equation}\label{eq:uniform-convergence}
\max_{1\leq i\leq n}\left|H^\LRU_{n,i} - H^\TTL_{n,i}(T_n)\right| \to 0, \quad\hbox{as }\ n\to\infty.
\end{equation}
\end{proposition}

\begin{remark}\label{rmk:mu}
Condition \eqref{ass:P1} requires that the popularity distribution $\bar P_n(i)$  take values of the same order for $i$ around $C_n$, as alluded to in \prettyref{subsec:popularity}. Intuitively, this means $\bar P_n(i)$ should not change abruptly around $i=C_n$. In a stronger form obtained by setting $\kappa_2=0$, \eqref{eq:heavy-tail} reads $\bar P_n(\lceil \kappa_1 C_n\rceil) > \gamma$, which means that even with a slightly larger cache, the contents that cannot fit into the cache have an aggregate probability at least $\gamma$, or equivalently, the optimal static caching policy has a miss probability at least $\gamma$. For Zipfian popularity in \eqref{eq:Zipf}, this stronger form is satisfied only for $\alpha \leq 1$, while \eqref{eq:heavy-tail} is satisfied for all $\alpha\geq 0$ as shown in \prettyref{ex:heavy-tail-Zipf}.
\end{remark}

\begin{example}\label{ex:heavy-tail-Zipf}
Consider the Zipfian popularity distribution in \eqref{eq:Zipf}.  We first check that assumption \eqref{ass:P1} is satisfied for all $\alpha\geq 0$. For large $n$,
\[
\bar P_n(i) \sim 
\begin{cases}
\frac{n^{1-\alpha} - i^{1-\alpha}}{n^{1-\alpha}}, &\text{if } 0\leq \alpha <1;\\[5pt]
\frac{\log n - \log i}{\log n}, &\text{if } \alpha = 1;\\[6pt]
\frac{i^{1-\alpha}}{(1-\alpha)\zeta(\alpha)}, &\text{if } \alpha > 1.
\end{cases}
\]
Thus 
\[
\liminf_{n\to\infty} \frac{\bar P_n(\lceil \kappa_1 C_n\rceil)}{\bar P_n(\lfloor \kappa_2 C_n\rfloor)} \geq
\begin{cases}
1-(\kappa_1 \beta_1)^{1-\alpha}, &\text{if } 0\leq \alpha <1;\\
1, &\text{if } \alpha = 1;\\
\left(\frac{\kappa_2}{\kappa_1}\right)^{\alpha-1}, &\text{if } \alpha > 1.\\
\end{cases}
\]
In all cases, the above guarantees the existence of a $\gamma\in (0,1)$ for which \eqref{eq:heavy-tail} holds. Note that for $\alpha\leq 1$, we can set $\kappa_2=0$.  If $m_\Psi = 1$, then $\cG = \{\Psi\}$, which satisfies \eqref{ass:R2} by \prettyref{lem:equicontinuity}. Thus \eqref{eq:uniform-convergence} holds for any $C_n$ satisfying \eqref{ass:C1}. In particular, \eqref{eq:uniform-convergence} holds when all request processes are Poisson.
\end{example}

As indicated in Section \ref{subsec:CTA},   Hirade and Osogami  proved in \cite{hirade2010analysis} that the individual content hit probability in an LRU cache converges to the individual content hit probability in a TTL cache for Poisson requests as the number of contents increases to infinity. More precisely, they consider $nN$ contents, $e_{i,j}$, $i=1,\ldots,N$, $j=1,\ldots,n$, each of size $1/n$,
where successive requests for content $e_{i,j}$ follow a Poisson process with rate $\lambda_i$. These Poisson processes are assumed to be mutually independent.
Note that in this setting there is only a finite number of types of requests ($=N$)\footnote{This can be considered as a special case of the setting in \prettyref{prop:H-limit-J} with $N$ classes, each consisting of $n$ equally popular contents. However, Theorem 1 of \cite{hirade2010analysis} concerns hit probabilities of individual contents, while \prettyref{prop:H-limit-J} concerns average hit probability.}. Define $F_i(t)=1-\exp(-\lambda_i t)$. It is shown in \cite[Theorem 1]{hirade2010analysis} that the probability, $p^{(n)}_{i,j}$, that content $e_{i,j}$ is in an LRU cache converges to $F_i(T)$ as $n\to\infty$, where $T$ is the unique solution of the equation $\sum_{i=1}^N F_i(T)=K$, with $K<N$ being the size of the cache. 
By performing the substitutions  $n\to nN$,  $G_{nN, (n-1)i+j}(\cdot)\to F_i(\cdot)$ for $j=1,\ldots,n$, $i=1,\ldots,N$ and $C_n\to nK$ (with these substitutions the ratio ``cache size/content size $=nK$'' is the same as in \cite{hirade2010analysis}), we get from \eqref{eq:convergence},
\[
H^{\LRU}_{nN,i}\sim H^{\TTL}_{nN,i}(T_{nN}) = F_i (T_{nN}) \quad \hbox{as }n\to\infty,
\]
where (see (\ref{eq:CT})) $T_{nN}$ is the unique $t$ satisfying the equation $Kn=\sum_{i=1}^{nN} \hat{G}_{nN,i}(t)=\sum_{i=1}^N nF_i(t)$, or equivalently, $K=\sum_{i=1}^N F_i(t)$.
We now check the conditions \eqref{ass:C1}, \eqref{ass:R1} and \eqref{ass:P1} for  \eqref{eq:convergence}.  
Condition \eqref{ass:C1} reads $nK\leq \beta_1 nN$,  which holds for any $K/N\leq \beta_1 <1$ (note that $m_{\Psi}=1$ since requests are Poisson).  By \prettyref{lem:equicontinuity}, $\cG_i=\{\Psi\}$ with $\Psi(t)=1-e^{-t}$ is equicontinuous, satisfying \eqref{ass:R1}.  To check
\eqref{ass:P1}, we first observe that contents $e_{i,1},\ldots,e_{i,n}$ have the same popularity $r_i/n\in (0,1)$ with $\sum_{i=1}^N r_i=1$.  Hence, 
$\bar  P_{nN}(\lceil \kappa_1 C_n \rceil)  \gtrsim N(1- \kappa_1\beta_1)\min_{1\leq i\leq N} r_i:= \gamma$. Since one  can find $\kappa_1 \in (1,1/\beta_1)$ such that $\gamma\in (0,1)$,
we have shown that \eqref{eq:heavy-tail} holds with this $\gamma$ and $\kappa_2=0$.

Note that a similar  fluid approximation for an LRU cache is developed in \cite{osogami2010fluid}, which considers dependent and so-called time-asymptotically stationary requests. However, the modification introduced to deal with the dependence structure renders the new approximation unsuitable for a re-interpretation as above.  Thus the results therein do not apply to TTL approximations.  Observe also that there is only empirical evidence but no theoretical proof that the fluid limit is an accurate approximation of the original LRU cache.

The following corollary considers the convergence of the aggregate hit probability.

\begin{corollary}\label{corollary:mean-convergence}
Assume \eqref{ass:C1} and \eqref{ass:P1}. Then as $n\to\infty$,  
\begin{equation}\label{eq:mean-convergence}
 \left|H_n^\LRU - H_n^\TTL(T_n)\right| \to 0,
 \end{equation}
if either \eqref{ass:R2} holds, or for each $i$, \eqref{ass:R1} and the following hold
\begin{equation}\label{eq:light-tail}
\lim_{m\to\infty}\limsup_{n\to\infty} \bar P_n(m) = 0.
\end{equation}
\end{corollary}

\begin{proof}
By \eqref{eq:H-H_i-LRU} and \eqref{eq:H-H_i-TTL}, for any $m$,
\[
\left|H_n^\LRU - H_n^\TTL(T_n)\right| \leq \max_{1\leq i\leq m} \left|H^\LRU_{n,i} - H^\TTL_{n,i}(T_n)\right| + \bar P_n(m).
\]

Suppose \eqref{ass:R2} holds. Let $m = n$. Since $\bar P_n(n) = 0$, \eqref{eq:mean-convergence} follows from \eqref{eq:uniform-convergence}.

Suppose for each $i$, \eqref{ass:R1} and \eqref{eq:light-tail} hold. Fix $m$ and let $n\to\infty$. By \eqref{eq:convergence},
\[
\limsup_{n\to\infty} \left|H_n^\LRU - H_n^\TTL(T_n)\right| \leq  \limsup_{n\to\infty} \bar P_n(m).
\]
Now let $m\to\infty$ and \eqref{eq:mean-convergence} follows \eqref{eq:light-tail}.
\end{proof}

\begin{example}
For the Zipfian popularity distribution in \eqref{eq:Zipf},
\[
\limsup_{n\to\infty} \bar P_n(m) = 
\begin{cases}
1, &\text{if } 0\leq \alpha \leq 1;\\
\frac{m^{1-\alpha}}{(1-\alpha)\zeta(\alpha)}, &\text{if } \alpha > 1.\\
\end{cases}
\]
Thus \eqref{eq:light-tail} holds for $\alpha>1$ but fails for $\alpha\in [0,1]$. For each $i$, if the standardized cdf $G^\ast_{n,i}$ is the same for all $n$, then $\cG_i$ is a singleton and hence equicontinuous by \prettyref{lem:equicontinuity}. In this case,  \eqref{eq:mean-convergence} holds for $\alpha>1$, but we cannot conclude the same for $\alpha\leq 1$ without further assuming that $\cG$ is equicontinuous. When $m_\Psi=1$, in particular, when all request processes are Poisson, \eqref{eq:mean-convergence} holds. For Poisson requests, Fagin \cite{fagin1977asymptotic} has established the convergence for $\alpha\in (0,1)$ and $C_n\sim \beta_0 n$. We now see this is also true for $\alpha\geq 1$ and for $C_n$ scaling sublinearly in $n$. 
\end{example}

\subsection{Proof of \prettyref{prop:T_n-bounds}}\label{subsec:proof-T_n-bounds}

We need the following two simple lemmas.

\begin{lemma}
\label{lem:Y}
For $i,j=1,\ldots,n$, and $t>0$,
\begin{equation} \label{eq:Y_i-conditional-pmf}
\P^0_{n,j}[Y_{n,i}(t) = 1] = \ind{\{j=i\}} G_{n,i}(t) + \ind{\{j\neq i\}} \hat G_{n,i}(t).
\end{equation}
\end{lemma}

\begin{proof}  For $i=j$, since $t_{n,i}(0)=0$ a.s. under $\P_{n,i}^0$, we have
\[
\P^0_{n,i}[Y_{n,i}(t) = 1]=\P^0_{n,i}[-t_{n,i}(-1) \leq  t] =G_{n,i}(t).
\]
For $i\neq j$, the independence of the point processes $N_{n,i}$ and $N_{n,j}$ yields
\[
\P^0_{n,j}[Y_{n,i}(t) = 1] = \P[Y_{n,i}(t) = 1];
\]
see \cite[Eq.~(1.4.5)]{BB-book-2003} for a more formal statement. 
Since $t_{n,i}(0) < 0$ a.s. under $\P$, we obtain 
\begin{equation}\label{eq:Y_i-pmf}
\P^0_{n,j}[Y_{n,i}(t) = 1] = \P[Y_{n,i}(t) = 1]= \P[-t_{n,i}(0)\leq t] = \hat G_{n,i}(t),
\end{equation}
where the last equality follows from \eqref{eq:age}.
This completes the proof of \eqref{eq:Y_i-conditional-pmf}.
\end{proof}

\begin{lemma} \label{lem:K}
The function $K_n$ defined in \eqref{eq:K} satisfies the following,
\begin{align} 
K_n(T) &= \sum_{i=1}^n \hat G_{n,i}(T), \label{eq:K-T}\\
K'_n(T) & = \sum_{i=1}^n \lambda_{n,i} \bar G_{n,i}(T). \label{eq:K-prime}
\end{align}
The function $K_n$ is concave on $[0,\infty)$ and strictly increasing at all $T\in [0,\infty)$ such that $K_n(T)<n$.
\end{lemma}

\begin{proof}
Using \eqref{eq:K}, \eqref{eq:Y} and \eqref{eq:Y_i-pmf}, we obtain
\[
K_n(T)=\E[Y_n(T)]=\sum_{i=1}^n  \P[Y_{n,i}(T) = 1] =\sum_{i=1}^n \hat G_{n,i}(T),
\]
proving \eqref{eq:K-T}.
Taking the derivative of \eqref{eq:K-T} w.r.t.~$T$ and using \eqref{eq:age} yield \eqref{eq:K-prime}.
Note that $K'_n$ is a decreasing function of $T$, from which it follows that $K_n(T)$ is concave.

Now we show that $K'_n(T)>0$ at all $T$ such that $K_n(T) < n$, from which it will follow that $K_n$ is strictly increasing at all such $T$.
Clearly $K_n^\prime (T)\geq 0$ from \eqref{eq:K-prime}.
Assume that $K_n^\prime(T)=0$ for some $T>0$. Then, $\bar G_{n,i}(T)=0$ for all $i$, which, by monotonicity of $\bar G_{n,i}$, yields $\bar G_{n,i}(y) = 0$ for all $y\geq T$. Thus, by \eqref{eq:age},
\[
1-\hat G_{n,i}(T) = \lambda_i \int_T^\infty \bar G_{n,i}(y) dy = 0,
\]
which implies $K_n(T)=n$ by \eqref{eq:K-T}. Therefore, $K_n^\prime(T)>0$  for all $T$ such that $K_n(T)<n$.
\end{proof}

Now we prove \prettyref{prop:T_n-bounds}.

\begin{proof}[Proof of \prettyref{prop:T_n-bounds}]
The existence of $T_n$ follows from the continuity of $K_n$, the facts $K_n(0) = 0$ and $\lim_{T\to\infty} K_n(T) = n$, and the Intermediate Value Theorem. Uniqueness follows from the strict monotonicity of $K_n$ given by \prettyref{lem:K}.

By \eqref{eq:age} and the fact $\bar G_{n,i}(y)\leq 1$, we have
\[
\hat G_{n,i}(T_n) =\lambda_{n,i} \int_0^{T_n} \bar G_{n,i}(y) dy \leq \lambda_{n,i} T_n.
\]
Thus
\begin{align*}
C_n = K_n(T_n) &= \sum_{i=1}^n \hat G_{n,i}(T_n)  \leq  \sum_{i=1}^n \min\{1,\lambda_{n,i} T_n\} \leq  \sum_{i=1}^{n_2} 1 +  \sum_{i=n_2+1}^n  \lambda_{n,i} T_n = n_2 + \Lambda_n T_n\bar P_n(n_2),
\end{align*}
from which \eqref{eq:T_n-lower} follows.

To prove \eqref{eq:T_n-upper}, note that
\[
C_n= \sum_{i=1}^n \hat G_{n,i} (T_n) \geq  m_\Psi \sum_{i=1}^n \hat \Psi(\lambda_{n,i} T_n) \geq n_1 m_\Psi\hat \Psi(\lambda_{n,\sigma_{n_1}} T_n),
\]
where the first inequality follows from \eqref{eq:Psi-hat}, and the second  from  \eqref{eq:p_i}, \eqref{eq:p-decreasing}, and the monotonicity of  $\hat\Psi$.
Since $C_n/(n_1 m_\Psi) < 1$ and $\hat\Psi$ is a continuous cdf, there exists a $\nu_0$ such that 
\[
\hat \Psi(\nu_0) \geq \frac{C_n}{n_1 m_\Psi}.
\] 
For any such $\nu_0$,
\[
\hat\Psi(\lambda_{n,\sigma_{n_1}} T_n) \leq \frac{C_n}{n_1 m_\Psi}\leq \hat\Psi(\nu_0),
\]
which, together with the monotonicity of $\hat \Psi$, yields \eqref{eq:T_n-upper}.
\end{proof}

\subsection{Proof of \prettyref{prop:convergence}}\label{subsec:proof-convergence}

The proof of Proposition \ref{prop:convergence} relies on the four lemmas  below.

Note by (\ref{eq:K-prime}) that $K_n'(T)$ is the aggregate miss rate of a TTL cache with timer $T$, and
\begin{equation}\label{eq:mu}
\mu_n(T) = \frac{K'_n(T)}{\Lambda_n} =\sum_{i=1}^n p_{n,i} \bar G_{n,i}(T)
\end{equation}
is the aggregate miss probability.

\begin{lemma}
\label{lem:mu}
Assume \eqref{ass:C1} and \eqref{ass:P1}. 
Then there exist strictly positive constants $x_0, \phi$ that do not depend on $n$,
such that for $T\leq (1+ x_0)T_n$ and sufficiently large $n$,
\begin{equation}\label{eq:mu-lower}
\mu_n(T) \geq \frac{ \phi  C_n }{ \Lambda_n T_n }.
\end{equation}
\end{lemma}

  \begin{proof}
 Recall the definition of $\kappa_1$ and $\kappa_2$ in the statement of Proposition \ref{prop:convergence}.
 Let $n_1 =\lceil \kappa_1 C_n\rceil$,  $n_2 =\lfloor \kappa_2 C_n\rfloor$. As  $C_n/(n m_\Psi) \leq \beta_1/ m_\Psi< 1$  for sufficiently large $n$ by \eqref{ass:C1}
 and $\hat \Psi$ is a
continuous cdf with $\hat \Psi(0) = 0$, there exist $\nu_0$ and $x_0 > 0$  such that
 \begin{equation}
 \label{bound-psi-hat}
 1> \hat \Psi((1+x_0)\nu_0) \geq \hat \Psi(\nu_0) \geq \beta_1 /m_\Psi \geq  C_n/(n m_\Psi)
 \end{equation}
 for sufficiently large $n$. Recall the content ordering \eqref{eq:p-decreasing}. For sufficiently large $n$,
\begin{align}
 \mu_n(T) &=\sum_{i=1}^n p_{n,\sigma_i} \bar G_{n,\sigma_i}( T) \nonumber\\
 &\geq \sum_{i=1}^n p_{n,\sigma_i} \bar \Psi(\lambda_{n,\sigma_{i} }T)\hspace{12mm} \text{by \eqref{eq:Psi-bar}}\nonumber\\
 &\geq \sum_{i=n_1+1}^n  p_{n,\sigma_i} \bar \Psi(\lambda_{n,\sigma_{i} }T) \nonumber\\
 &\geq \bar \Psi(\lambda_{n,\sigma_{n_1} }T)\sum_{i=n_1+1}^n p_{n,\sigma_i} \hspace{4mm}\text{by \eqref{eq:p-decreasing}}\nonumber\\
 &=\bar \Psi( \lambda_{n,\sigma_{n_1}} T) \bar P_n(n_1).
 \label{mun-T}
 \end{align}
 Since $\mu_n(T)$ is monotonically decreasing in $T$, we obtain for $T\leq (1+x_0)T_n$ and all sufficiently large $n$,
 \begin{align*}
 \mu_n(T) &\geq \mu_n((1+x_0)T_n)\\
 &\geq \bar \Psi((1+x_0)\lambda_{n,\sigma_{n_1}} T_n) \bar P_n(n_1)\hspace{8.7mm}  \text{by (\ref{mun-T})} \\
 &\geq \bar \Psi((1+x_0)\nu_0) \bar P_n(n_1) \hspace{17.8mm} \text{by \eqref{eq:T_n-upper}}\\
 & \geq \frac{C_n-n_2}{\Lambda_n T_n} \bar \Psi((1+x_0)\nu_0) \frac{\bar P_n(n_1)}{\bar P_n(n_2)} \hspace{5.1mm} \text{by \eqref{eq:T_n-lower}}\\
 &\geq \frac{(1-\kappa_2)  C_n}{\Lambda_n T_n}   \bar \Psi((1+x_0)\nu_0)  \frac{\bar P_n(n_1)}{\bar P_n(n_2)} \nonumber\\
 & \geq  \frac{(1-\kappa_2)  C_n}{\Lambda_n T_n}  \bar \Psi((1+x_0)\nu_0) \gamma \hspace{9mm} \text{by \eqref{eq:heavy-tail}}.
 \end{align*}
 The last inequality yields \eqref{eq:mu-lower} with $\phi = (1-\kappa_2) \gamma \bar \Psi((1+x_0) \nu_0)$ if $\bar \Psi((1+x_0) \nu_0)>0$.
 Assume that $\bar \Psi((1+x_0) \nu_0)=0$. This would imply that $\Psi(x)=1$ for all $x\geq (1+x_0) \nu_0$ by monotonicity of $\Psi$, which would in turn imply that
 $1-\hat \Psi((1+x_0)\nu_0)=(1/m_\Psi)\int_{(1+x_0)\nu_0}^\infty \bar \Psi(t)dt =0$, contradicting \eqref{bound-psi-hat}.
Therefore, we indeed have
 $\bar \Psi((1+x_0) \nu_0)>0$, which completes the proof.
 \end{proof}

\begin{lemma}[Kolmogorov's inequality \hbox{\cite[Section 19.1]{loeve1977probability}}]
Let $X_1,\dots,X_n$ be independent random variables such that $\E X_i = 0$ and $|X_i|\leq b$ for all $i$. Then for any $x> 0$,
\begin{equation}\label{eq:Kolmogorov}
\P\left[\sum_{i=1}^n X_i \geq x\right] \leq \exp\left\{-\frac{x^2}{4 \max\{s_n^2, bx\}}\right\},
\end{equation}
where $s_n^2 = \sum_{i=1}^n \E X_i^2$ is the variance of $\sum_{i=1}^n X_i$. 
\end{lemma}

The next lemma shows that $\tau_n$ is concentrated around $T_n$.

\begin{lemma}
\label{lemma:tau-concentration}
Assume \eqref{eq:mu-lower} holds for $T\leq (1+x_0)T_n$. Then for $0\leq x\leq \min\{1,x_0\}$, 
\[
\P^0_{n,i}[\tau_n > (1+x)T_n]  \leq \exp\left\{-\frac{(\phi xC_n)^2}{4(1+x)C_n + 4}\right\}.
\]
If, in addition, $\phi x C_n \geq 1$, then
\[
\P^0_{n,i}[\tau_n < (1-x)T_n] \leq \exp\left\{-\frac{(\phi xC_n-1)^2}{4C_n + 4}\right\}.
\]
\end{lemma}
\begin{proof}
Let $T_n^+ = (1+x) T_n$ and $T_n^- = (1-x)T_n$. 
Note that
\[
K_n(T_n^+) - C_n = K_n(T_n^+) - K_n(T_n) = \int_{T_n}^{T_n^+} K'_n(T) dT,
\]
which, by \eqref{eq:mu} and \eqref{eq:mu-lower}, yields
\begin{equation}
\label{eq:K(T_n^+)-C}
K_n(T_n^+) - C_n \geq \int_{T_n}^{T_n^+}  \frac{\phi C_n}{T_n} dT =  \phi x C_n.
\end{equation}
Since $T_n = (T_n^+ + T_n^-)/2$, the concavity of $K_n$ yields
\begin{equation}\label{eq:C-K(T_n^-)}
C_n - K_n(T_n^-)  = K_n(T_n) - K_n(T_n^-) \geq  K_n(T_n^+) - K_n(T_n)  = K_n(T_n^+) - C_n \geq \phi x C_n
\end{equation}
by (\ref{eq:K(T_n^+)-C}).
Note that by \eqref{eq:Y}, \eqref{eq:Y_i-conditional-pmf} and \prettyref{eq:K-T}, we have
\[
\E^0_{n,i}[Y_n(T)] = \sum_{j=1}^n \E^0_{n,i}[Y_{n,j}(T)] = K_n(T) + G_{n,i}(T) - \hat G_{n,i}(T).
\]
Since $G_{n,i}$ and $\hat G_{n,i}$ are both cdfs, we obtain
\begin{equation}\label{eq:E^0_{n,i}(Y)}
K_n(T)-1 \leq \E^0_{n,i}[Y_n(T)] \leq K_n(T)+1.
\end{equation}
Using the definition of $\tau_n$ in \eqref{eq:tau}, we obtain
\begin{align*}
\P^0_{n,i}[\tau_n > T_n^+] &= \P_{n,i}^0[Y_n(T_n^+) \leq  C_n - 1]  = \P^0_{n,i}\left\{Y_n(T_n^+)- \E^0_{n,i}[Y_n(T_n^+)]  
\leq   C_n - 1 - \E^0_{n,i}[Y_n(T_n^+)]\right\}.
\end{align*}
By \eqref{eq:E^0_{n,i}(Y)} and \eqref{eq:K(T_n^+)-C},
\[
C_n - 1 - \E^0_{n,i}[Y_n(T_n^+)] \leq C_n - K_n(T_n^+) \leq -\phi x C_n.
\]
Thus
\begin{equation}\label{eq:tau-T-bound}
\P^0_{n,i}[\tau_n > T_n^+]    \leq \P^0_{n,i}\left[Y_n(T_n^+)- \E^0_{n,i}[Y_n(T_n^+)]  \leq   - \phi x C_n\right].
\end{equation}
Since the request processes $N_{n,1}, N_{n,2}\dots, N_{n,n}$ are independent,  so are the Bernoulli random variables $Y_{n,1}(t), Y_{n,2}(t), \ldots, Y_{n,n}(t)$ under $\P_{n,i}^0$. Thus
\begin{align}
\var_{n,i}^0[Y_n(T_n^+)] &= \sum_{j=1}^n \var_{n,i}^0[Y_{n,j}(T_n^+)]\leq \sum_{j=1}^n \E_{n,i}^0 [Y_{n,j}(T_n^+)] = \E_{n,i}^0 [Y_{n}(T_n^+)] \nonumber\\
&\leq K_n(T_n^+) + 1 \quad\hbox{by }(\ref{eq:E^0_{n,i}(Y)})\nonumber\\
& \leq (1+x)C_n + 1,
\end{align}
where last step follows from the following consequence of the concavity of $K_n$
\[
\frac{x}{1+x} K_n(0) + \frac{1}{1+x} K_n(T_n^+) \leq K_n\left(\frac{T_n^+}{1+x}\right) = K_n(T_n) = C_n
\]
and the fact $K_n(0) = 0$.\\

Note that $\left|Y_{n,i}(T)-\E_{n,i}^0[Y_{n,i}]\right| \leq 1$.  
By applying Kolmogorov's inequality \eqref{eq:Kolmogorov} with $b=1$ and $s^2_n\leq (1+x)C_n+1$ to the r.h.s.~of \eqref{eq:tau-T-bound}, we obtain
\[
\P^0_{n,i}[\tau_n > T_n^+] \leq \exp\left\{-\frac{(\phi xC_n)^2}{4(1+x)C_n + 4}\right\}.
\]
Similarly, if $\phi x C_n \geq 1$, we have
\begin{align*}
\P^0_{n,i}[\tau_n < T_n^-] &= \P^0_{n,i}[Y_n(T_n^-) \geq C_n]\\
& = \P^0_{n,i}\left[Y_n(T_n^-) - \E^0_{n,i}[Y_n(T_n^-)] \geq C_n - \E^0_{n,i}[Y_n(T_n^-)]\right]\\
&\leq \P^0_{n,i}\left[Y_n(T_n^-) - \E^0_{n,i}[Y_n(T_n^-)] \geq  C_n - K_n(T_n^-)-1\right]\quad \hbox{by  \eqref{eq:E^0_{n,i}(Y)}}\\
&\leq \P^0_{n,i}\left[Y_n(T_n^-) - \E^0_{n,i}[Y_n(T_n^-)] \geq  \phi x C_n - 1\right]\quad \hbox{by  \eqref{eq:C-K(T_n^-)}}\\
&\leq \exp\left\{-\frac{(\phi xC_n-1)^2}{4C_n + 4}\right\}
\quad \hbox{by \eqref{eq:Kolmogorov}}.
\end{align*}
\end{proof}

\begin{lemma}
\label{lem:inq}
\[
\P^0_{n,i}[Y_{n,i}(\tau_n) = 1, \tau_n \leq T] \leq \P^0_{n,i}[Y_{n,i}(T) = 1, \tau_n \leq T],
\]
and
\[
\P^0_{n,i}[Y_{n,i}(\tau_n) = 1, \tau_n \geq T] \geq \P^0_{n,i}[Y_{n,i}(T) = 1, \tau_n \geq T].
\]
\end{lemma}

\begin{proof}
Since $Y_{n,i}(t)$ is increasing in $t$, the inequalities follow from a sample path argument.
\end{proof}

Now we prove \prettyref{prop:convergence}.
 
\begin{proof}[Proof of \prettyref{prop:convergence}]
Fix an arbitrary $\epsilon > 0$. We show that for large enough $n$,
\begin{equation}\label{eq:H-epsilon}
\left|H_{n,i}^\LRU - H_{n,i}^\TTL(T_n)\right| \leq 2\epsilon.
\end{equation}
The proof consists of two steps. We first show that $H_{n,i}^\TTL(T_n)$ is within $\epsilon$ distance from both $H_{n,i}^\TTL(T^+_n)$ and $H_{n,i}^\TTL(T^-_n)$ for some $T^+_n$ and $T^-_n$ to be defined below. We then show that $H_{n,i}^\LRU$ is within $\epsilon$ distance from at least one of $H_{n,i}^\TTL(T^+_n)$ and $H_{n,i}^\TTL(T^-_n)$.

Let $x_0$ and $\phi$ be given by \prettyref{lem:mu}. 
Since the family $\cG_i$ is equicontinuous by \eqref{ass:R1}, there exists $\xi_i(\epsilon)>0$ such that $|t_1 - t_2| \leq \xi_i(\epsilon)$ implies $|G^\ast_{n,i}(t_1) -  G^\ast_{n,i}(t_2)|\leq \epsilon$. Since $C_n\to\infty$ as $n\to\infty$, let $n$ be sufficiently large so that 
\[
C_n \geq \max\left\{\frac{1}{\phi x_0}, \frac{1+\epsilon \xi_i(\epsilon)}{\phi \epsilon \xi_i(\epsilon)}\right\},
\]
which guarantees the existence of an $x$ satisfying the following,
\begin{equation}\label{eq:x-bound}
\frac{1}{\phi C_n}\leq x\leq \min\left\{x_0, \frac{\epsilon  \xi_i(\epsilon)}{1+\epsilon \xi_i(\epsilon)}\right\}.
\end{equation}
Fix such an $x$. Let $T_n^+ = (1+x) T_n$, $T_n^-  = (1-x) T_n$.

We first show 
\begin{equation}\label{eq:H_T_n-lower}
H^\TTL_{n,i}(T_n) - H^\TTL_{n,i}(T_n^-)\leq \epsilon,
\end{equation}
and
\begin{equation}\label{eq:H_T_n-upper}
H^\TTL_{n,i}(T_n^+) - H^\TTL_{n,i}(T_n)\leq \epsilon.
\end{equation}
We only shown \eqref{eq:H_T_n-lower}, as \eqref{eq:H_T_n-upper} follows from the same argument. By \prettyref{lem:Y}, \eqref{eq:H_T_n-lower} is the same as $G_{n,i}(T_n) - G_{n,i}(T_n^-)\leq \epsilon$.
Note that (this result holds regardless of the values of $\epsilon$, $\xi_i(\epsilon)$ and $\lambda_{n,i}T_n$)
\[
\max\left\{1-\frac{1}{\epsilon \lambda_{n,i} T_n}, \frac{\xi_i(\epsilon)}{\lambda_{n,i} T_n}\right\}\geq \frac{\epsilon \xi_i(\epsilon)}{1+\epsilon \xi_i(\epsilon)}.
\]
Since $x$ satisfies \eqref{eq:x-bound}, there are two cases: either $x\leq \xi_i(\epsilon)/ (\lambda_{n,i} T_n)$ or $x\leq 1- (\epsilon\lambda_{n,i} T_n)^{-1}$. 
In the first case, $|\lambda_{n,i}T_n -  \lambda_{n,i}T_n^-| = x \lambda_{n,i}T_n \leq \xi_i(\epsilon)$. Since $G_{n,i}(t) = G^\ast_{n,i}(\lambda_{n,i} t)$, using the definition of $\xi_i(\epsilon)$, we obtain $G_{n,i}(T_n) - G_{n,i}(T_n^-)\leq \epsilon$.
In the second case, note that 
\[
G_{n,i}(T_n) - G_{n,i}(T_n^-)\leq 1-G_{n,i}(T_n^-)=   \bar G_{n,i}(T_n^-),
\]
and
\[
1/\lambda_{n,i} = \int_0^\infty \bar G_{n,i}(y)dy \geq \int_0^{T^-_n} \bar G_{n,i}(y)dy \geq T^-_n \bar G_{n,i}(T^-_n).
\]
Thus
\[
G_{n,i}(T_n) - G_{n,i}(T_n^-) \leq \bar G_{n,i}(T_n^-) \leq \frac{1}{\lambda_{n,i} T_n^-}\leq \epsilon,
\]
where the last inequality follows from the definition $T^-_n=(1-x)T_n$ and the condition $x\leq 1- (\epsilon\lambda_{n,i} T_n)^{-1}$. This proves \eqref{eq:H_T_n-lower}.

 Next we show \eqref{eq:H-epsilon}.
 By \prettyref{lemma:tau-concentration}, for sufficiently large $C_n$,
 \begin{equation}\label{eq:T-tail}
\left. \begin{aligned}
 \P^0_{n,i}[\tau_n > T_n^+]\\
 \P^0_{n,i}[\tau_n < T_n^-]
 \end{aligned} \right\} \leq \exp\left\{-\frac{(\phi xC_n-1)^2}{4(1+x)C_n + 4}\right\} \leq \epsilon.
\end{equation}

 Note that
\begin{align*}
H^\LRU_{n,i} &=\P^0_{n,i}[Y_{n,i}(\tau_n) = 1] \\
&\geq \P^0_{n,i}[Y_{n,i}(\tau_n) = 1, \tau_n \geq T_n^-] \\
&\geq \P^0_{n,i}[Y_{n,i}(T_n^-) = 1, \tau_n \geq T_n^-]\quad  \hbox{by Lemma \ref{lem:inq}} \\
&\geq \P^0_{n,i}[Y_{n,i}(T_n^-) = 1] - \P^0_{n,i}[\tau_n < T_n^-]\\
& = H^\TTL_{n,i}(T_n^-) - \P^0_{n,i}[\tau_n < T_n^-],
\end{align*}
which, by \eqref{eq:H_T_n-lower} and \eqref{eq:T-tail}, yields
\begin{align*}
 H_{n,i}^\TTL(T_n) - H_{n,i}^\LRU  \leq H^\TTL_{n,i}(T_n) - H_{n,i}^\TTL(T_n^-)  + \P^0_{n,i}[\tau_n < T_n^-] \leq 2\epsilon.
\end{align*}
Note that similar bounds have been used for the shot noise model in \cite{leonardi2017modeling}.

For the other direction, note that
\begin{align*}
H^\LRU_{n,i} &\leq \P^0_{n,i}[Y_{n,i}(\tau_n) = 1, \tau_n \leq T_n^+] + \P^0_{n,i}[\tau_n > T_n^+]\\
&\leq \P^0_{n,i}[Y_{n,i}(T_n^+)= 1, \tau_n \leq T_n^+] + \P^0_{n,i}[\tau_n > T_n^+] \quad \hbox{by \prettyref{lem:inq}}\\
&\leq \P^0_{n,i}[Y_{n,i}(T_n^+) = 1] + \P^0_{n,i}[\tau_n > T_n^+]\\
&=H^\TTL_{n,i}(T_n^+) + \P^0_{n,i}[\tau_n > T_n^+],
\end{align*}
which, by \eqref{eq:H_T_n-upper} and \eqref{eq:T-tail}, yields
\begin{align*}
H_{n,i}^\LRU - H_{n,i}^\TTL(T_n)  \leq H^\TTL_{n,i}(T_n^+) - H_{n,i}^\TTL(T_n)  + \P^0_{n,i}[\tau_n < T_n^+] \leq 2\epsilon.
\end{align*}
Therefore, \eqref{eq:H-epsilon} holds, which proves \prettyref{eq:convergence}.
 
Finally, \eqref{eq:uniform-convergence} follows from the same argument with $\xi_i(\epsilon)$ replaced by $\xi(\epsilon)$, whose existence is guaranteed by \eqref{ass:R2}, i.e. the equicontinuity of the family $\cG$. 
\end{proof}

\begin{remark}
\label{remark-prop.4-4}
In the above proof of \prettyref{prop:convergence}, the conditions \eqref{ass:C1} and \eqref{ass:P1} are used only to establish \eqref{eq:mu-lower} in Lemma \ref{lem:mu}.
Therefore,   \prettyref{prop:convergence} and Corollary \ref{corollary:mean-convergence} will still hold if \eqref{ass:C1} and \eqref{ass:P1}  are replaced by \eqref{eq:mu-lower} or other conditions that imply \eqref{eq:mu-lower}. 
\end{remark}

\begin{remark}
\label{remark:Fricker-et-al}
Note that \cite{fricker2012versatile} provides a more concise argument to justify the TTL approximation in the case of Poisson requests, but the argument does not constitute a rigorous proof of the asymptotic exactness of the approximation for this case. This is so for the following two reasons. First,
Proposition 2 therein assumes the quantity $X(t)$ is precisely Gaussian without investigating the error in this Gaussian approximation. Second, the analysis after Proposition 2 replaces the erfc function by the step function without further investigating the error introduced. 
\end{remark}


\section{Rate of Convergence }\label{sec:rate}

In this section, we provide two bounds on the rate of convergence in the TTL approxmation under different sets of assumptions.

The following proposition provides a convergence rate of order $(\log C_n/C_n)^{1/4}$. It is stated for the uniform convergence of hit probabilities assuming \eqref{ass:R5}, the uniform Lipschitz continuity of $\cG$. The obvious modification gives the convergence rate for content $i$ assuming uniform Lipschitz continuity of $\cG_{n,i}$. Examples of uniformly Lipschitz continuous cdfs include families of distributions that have densities with a common upper bound. 

\begin{proposition}\label{prop:rate}
Under assumptions \eqref{ass:C1}, \eqref{ass:R5} and \eqref{ass:P1}, the following holds,
\begin{equation}\label{eq:rate}
\max_{1\leq i\leq n}\left|H^\LRU_{n,i} - H^\TTL_{n,i}(T_n)\right|  = O\left(\left(\frac{\log C_n}{C_n}\right)^{\frac{1}{4}}\right).
\end{equation}
\end{proposition}

\begin{proof}
Let $M$ be the Lipschitz constant in \eqref{ass:R5}. By setting $\xi(\epsilon) = \epsilon/M$ in the proof of \prettyref{prop:convergence}, we obtain the following,
\[
\max_{1\leq i\leq n}\left|H^\LRU_{n,i} - H^\TTL_{n,i}(T_n)\right| \leq \epsilon + \exp\left\{-\frac{(\phi xC_n-1)^2}{4(1+x)C_n + 4}\right\},
\]
for $\frac{1}{\phi C_n}\leq x\leq \frac{\epsilon^2}{M+\epsilon^2}$. For fixed $x$, the smallest $\epsilon$ is $\epsilon = \sqrt{\frac{xM}{1-x}}$. Thus
\[
\max_{1\leq i\leq n}\left|H^\LRU_{n,i} - H^\TTL_{n,i}(T_n)\right| \leq \sqrt{\frac{xM}{1-x}} + \exp\left\{-\frac{(\phi xC_n-1)^2}{4(1+x)C_n + 4}\right\}.
\]
Let $x = \frac{1}{\phi}\sqrt{\frac{\log C_n}{ C_n}}$, which satisfies $\frac{1}{\phi C_n}\leq x\leq x_0$ when $C_n$ is large enough. Then the first term on the r.h.s.~of the above inequality is asymptotically equal to
\[
\sqrt{\frac{M}{\phi}}\left(\frac{\log C_n}{C_n}\right)^{\frac{1}{4}} = \Theta\left(\left(\frac{\log C_n}{C_n}\right)^{\frac{1}{4}}\right),
\]
while the second term is asymptotically equal to
\[
\exp\left\{-\frac{1}{4} \log C_n + o(1)\right\}\sim C_n^{-1/4}.
\]
It immediately follows that \eqref{eq:rate} holds.
\end{proof}

The next proposition provides a faster rate of convergence  under a different  condition, \eqref{ass:R6}, which says the change in the value of a cdf is bounded by a constant multiple of the relative change in its argument. In fact, we only need \eqref{eq:Lipschitz} to hold with $t=T_n$. Numerical results (see e.g \cite{fricker2012versatile}) show that the approximation may converge faster in practice than suggested by \eqref{eq:rate-quadratic}.

\begin{proposition}\label{prop:rate-faster}
Under assumptions \eqref{ass:C1}, \eqref{ass:R6} and \eqref{ass:P1}, the following holds,
\begin{equation}\label{eq:rate-quadratic}
\max_{1\leq i\leq n}\left|H^\LRU_{n,i} - H^\TTL_{n,i}(T_n)\right|  = O\left(\sqrt{\frac{\log C_n}{C_n}} \right).
\end{equation}
\end{proposition}

\begin{proof}
Note that the inequality in \eqref{eq:Lipschitz} is invariant under scaling of $t$, so \eqref{ass:R6} implies that \eqref{eq:Lipschitz} holds for $G_{n,i}, \forall n,i$. Replacing the bounds $G_{n,i}(T_n) - G_{n,i}(T_n^+)\leq \epsilon$ and $G_{n,i}(T_n) - G_{n,i}(T_n^-)\leq \epsilon$ by \eqref{eq:Lipschitz} in the proof of \prettyref{prop:convergence}, we obtain the following,
\[
\max_{1\leq i\leq n}\left|H^\LRU_{n,i} - H^\TTL_{n,i}(T_n)\right| \leq B x + \exp\left\{-\frac{(\phi xC_n-1)^2}{4(1+x)C_n + 4}\right\},
\]
for $\frac{1}{\phi C_n}\leq x\leq \min\{\rho,x_0\}$. Let $x = \frac{1}{\phi}\sqrt{\frac{2 \log C_n}{ C_n}}$, which falls in the interval $[(\phi C_n)^{-1}, \min\{\rho, x_0\}]$ when $C_n \geq \max\{2, (\min\{\rho, x_0\} \phi)^{-4}\}$. Then the second term on the r.h.s.~of the above inequality is asymptotically equal to
\[
\exp\left\{-\frac{1}{2} \log C_n + o(1)\right\}\sim C_n^{-1/2}.
\]
It immediately follows that \eqref{eq:rate-quadratic} holds.
\end{proof}

The following examples show that \eqref{ass:R6} holds for a large class of distributions. 

\begin{example}
\label{ex:Poisson}
For Poisson request processes, $\cG = \{\Psi\}$ with $\Psi(t) = 1-e^{-t}$. For any $x \geq 0$,
\begin{align*}
0\leq \Psi(t + x t) - \Psi(t) = e^{- t} (1-e^{- x t}) \leq x t e^{- t}  \leq e^{-1} x,
\end{align*}
where we have used inequalities $e^{-z}\geq 1-z$ and $ze^{-z} \leq e^{-1}$.
For $x \in [0,1]$,
\begin{align*}
0\leq  \Psi(t)-\Psi(t - x t)  \leq \sup_{z \geq 0} e^{-z}(e^{x z} - 1) = (1-x)^{\frac{1}{x}-1} x \leq x.
\end{align*} 
Thus \eqref{ass:R6} holds with $B = 1$ and $\rho = 1$.
\end{example}

\begin{example} \label{ex:general}
Suppose every $G\in \cG$ has continuous density on $(0,\infty)$. By the Mean Value Theorem, there exists $\xi_G \in [1,1+x]$ such that
\begin{align*}
0\leq  G(t + x t) - G(t) = G'(\xi_G t) x t \leq \xi_G t G'(\xi_G t) x    \leq \left[\sup_{t > 0} t  G'(t)\right] x \leq B_0 x,
\end{align*}
where
\[
B_0= \sup_{G\in\cG} \sup_{t > 0} t G'(t).
\]
Similarly, there exists $\zeta_G \in [1-x, 1]$ such that
\begin{align*}
0 \leq  G(t) - G(t - x t) = G'(\zeta_G t) x t \leq \frac{\zeta_G}{1-x}  t  G'(\zeta_G  t) x  \leq \frac{x}{1-x}\left[\sup_{t > 0} t G'(t)\right] \leq \frac{B_0}{1-x} x.
\end{align*}
If $B_0<\infty$, then \eqref{ass:R6} holds with any $\rho \in (0,1)$ and $B = \frac{B_0}{1-\rho}$. When is $B_0<\infty$ then?
Since $G$ has finite mean, $\sup_{t > 0} t G'(t)<\infty$. 
If $\cG$ is finite, i.e.~the $G_{n,i}$'s are from a finite number of scale families, then $B_0<\infty$ after taking the supremum over a finite set. 
In particular, for Poisson request processes, $\cG = \{\Psi\}$ with $\Psi(t) = 1-e^{-t}$, so
\[
B_0 = \sup_{t> 0} t \Psi'(t) = \sup_{t> 0} t e^{-t} = e^{-1} < \infty.
\]
Thus  \eqref{ass:R6} holds with any $\rho\in (0,1)$ and $B = e^{-1}(1-\rho)^{-1}$, which is weaker than what we have obtained in \prettyref{ex:Poisson}. 

However, when $\cG$ is infinite, i.e., the $G_{n,i}$'s are not from a finite number of scale families,  $B_0$ may still diverge to infinity when we take the supremum over $G\in \cG$. An example where we still have finite $B_0$ is provided by an infinite collection of gamma distributions with shape parameters upper bounded by some $\alpha_{\max}<\infty$. Recall that a gamma distribution $G_\alpha$ with unit mean and shape parameter $\alpha>0$ has the following density,
\[
G'_{\alpha}(t) = \frac{\alpha^\alpha}{\Gamma(\alpha)} t^{\alpha - 1} e^{-\alpha t}, \quad t>0.
\]
Hence
\[
\sup_{t > 0} t G_{\alpha}'(t) =  \frac{\alpha^\alpha}{\Gamma(\alpha)} \sup_{t> 0} t^{\alpha} e^{-\alpha t} = \frac{\alpha^\alpha}{\Gamma(\alpha)} \left(\sup_{t > 0} t e^{-t}\right)^\alpha = \frac{\alpha^{\alpha} e^{-\alpha}}{\Gamma(\alpha)},
\]
and
\[
B_0 =   \sup_{\alpha:G_\alpha\in \cG} \frac{\alpha^{\alpha} e^{-\alpha}}{\Gamma(\alpha)} \leq  \sup_{0< \alpha \leq \alpha_{\max}} \frac{\alpha^{\alpha} e^{-\alpha}}{\Gamma(\alpha)}.
\]
Since the function $\alpha^{\alpha} e^{-\alpha}/\Gamma(\alpha)$ is continuous and has limit $0$ as $\alpha\to 0$, we obtain $B_0<\infty$. Note that  as $\alpha\to\infty$, 
\[
\frac{\alpha^{\alpha} e^{-\alpha}}{\Gamma(\alpha)} \sim \sqrt{2\pi \alpha}\to\infty,
\]
so the boundedness of $\alpha$ is essential.
\end{example}

\begin{corollary}
Assume $C_n \leq \beta_1 n$ for some $\beta_1\in (0,1)$ and the popularity distribution is Zipf's law in \eqref{eq:Zipf}. Then \eqref{eq:rate-quadratic} holds if $m_\Psi=1$ and $\Psi$ has a continuous density. In particular,  \eqref{eq:rate-quadratic} holds if all request processes are Poisson.
\end{corollary}

\begin{proof}
We check the assumptions of \prettyref{prop:rate-faster}. Condition \eqref{ass:C1} is assumed.
Condition \eqref{ass:P1} holds for Zipfian popularity by \prettyref{ex:heavy-tail-Zipf}. By \prettyref{ex:general}, condition \eqref{ass:R6} holds when $m_\Psi=1$ and $\Psi$ has a continuous density.
\end{proof}


\section{Extension of Fagin's Result}\label{sec:Fagin}

In this section, we derive 
expressions for the characteristic time and the aggregate hit probability in the limit as the cache size and the number of contents go to infinity. This extends the results of Fagin \cite{fagin1977asymptotic} for the independence reference model to the more general setting of independent stationary and ergodic content request processes. 

We first consider the case where $m_\Psi = 1$ and $p_{n,i} \sim g_nf(z_{n,i})$ uniformly for some continuous function $f$ defined on $(0,1]$ and $z_{n,i} \in [\frac{i-1}{n},\frac{i}{n}]$, i.e. \eqref{ass:R4} and \eqref{ass:P2} hold. Recall that $m_\Psi = 1$ implies the cdfs $G_{n,i}$ are all from the same scale family, i.e.~$G_{n,i}(t) = \Psi(\lambda_{n,i} t)$ for all $n$ and $i$. 

The following proposition gives the asymptotic expression of $T_n$, which will be used in the proof of \prettyref{prop:H-limit} and is also of independent interest.
Note that \eqref{eq:T_n-asympt} is a generalization of Eq.~(2.2) of \cite{fagin1977asymptotic} and Eq.~(7) of \cite{fricker2012versatile}. We have imposed the inessential condition $f>0$ a.e.~on $[0,1]$, which simplifies the statements and can be easily removed. The proof is found in \prettyref{subsec:proof-T_n-asympt}. 

\begin{proposition}\label{prop:T_n-asympt}
Under assumptions \eqref{ass:C2}, \eqref{ass:R4} and \eqref{ass:P2}, the following holds
\begin{equation}\label{eq:T_n-asympt}
T_n \sim \frac{\nu_0}{g_n\Lambda_n},
\end{equation}
where $\nu_0$ the unique real number in $(0,\infty)$ that satisfies
\begin{equation}\label{eq:value-nu0-prop:T_n-asympt}
\int_0^1 \hat \Psi(\nu_0 f(x))dx = \beta_0.
\end{equation}
\end{proposition}

\begin{example}\label{ex:Fagin-Zipf}
Consider Zipf's law in \eqref{eq:Zipf} with $\alpha\geq 0$. Then $p_{n,i} \sim g_n f(i/n)$ with $f(x) = x^{-\alpha}$ and 
\[
g_n = \begin{cases}
\frac{1-\alpha}{n}, & \text{if } \alpha<1;\\[2pt]
\frac{1}{n\log n}, & \text{if } \alpha=1;\\[3pt]
\frac{1}{\zeta(\alpha)n^{\alpha}}, & \text{if } \alpha>1.
\end{cases}
\]
It is easy to check that
 \[
\max_{1\leq i \leq n} \left|\frac{g_nf(i/n)}{ p_{n,i}} - 1\right |= \left|g_nn^{\alpha}\sum_{j=1}^n j^{-\alpha} - 1\right |\to 0
 \]
as $n\to \infty$, so \eqref{ass:P2} holds. If \eqref{ass:C2} and \eqref{ass:R4} also hold, then $T_n$ satisfies \eqref{eq:T_n-asympt}. In particular, if $\lambda_{n,i} = i^{-\alpha}$, then $g_n\Lambda_n \sim n^{-\alpha}$ and hence $T_n \sim \nu_0 n^{\alpha}$. For Poisson request processes, $\hat \Psi(t) = 1-e^{-t}$ and we recover Eq.~(7) of \cite{fricker2012versatile}.
\end{example}

The following proposition gives the limiting aggregate hit probability, which generalizes Eq.~(2.3) of \cite{fagin1977asymptotic}. The proof is found in \prettyref{subsec:proof-H-limit}.

\begin{proposition}\label{prop:H-limit}
Assume \eqref{ass:C2}, \eqref{ass:R4} and \eqref{ass:P2} with $g_n = n^{-1}$.  Then,
\begin{equation}
\label{eq:H-limit}
H_n^\LRU \to \int_0^1 f(x) \Psi(\nu_0 f(x)) dx,
\end{equation}
as $n\to \infty$, where $\nu_0$ satisfies \eqref{eq:value-nu0-prop:T_n-asympt}.
\end{proposition}

\prettyref{prop:H-limit} considers a single class of contents in the sense that there is a single $f$ and a single $\Psi$ for all contents. Consider the following generalization to a setting with multiple classes of contents, which may arise from a situation where multiple service providers share a common LRU cache. More precisely, consider $J$ classes of contents, where class $j$ has $b_j n$ contents\footnote{We assume $b_j n$ is an integer for ease of presentation, but this can easily relaxed by requiring class $j$ to have a fraction $b_j$ of the contents asymptotically.} with $b_j>0$ and $\sum_{j=1}^J b_j=1$. Instead of labeling  contents by a single index $i$, we label them by a double index so that $(j,k)$ is the $k$-th content belonging to class $j$. Correspondingly, we have $\lambda_{n,j,k}$ instead of $\lambda_{n,i}$, and similarly for other quantities. For each class $j$, 
\begin{enumerate}[label=(\alph*), ref=\alph*]

\item\label{ass:a} the inter-request distributions are from the same scale family, i.e. $G_{n,j,k}(x) = \Psi_j(\lambda_{n,j,k} x)$ for some continuous cdf $\Psi_j$ with support in $[0,\infty)$ and $m_{\Psi_j}=1$;

\item\label{ass:b} the content popularities $p_{n,j,k} \sim n^{-1} f_j(z_{n,j,k})$ uniformly in $k$ for $z_{n,j,k}\in [\frac{k-1}{b_j n},\frac{k}{b_j n}]$ and  continuous function $f_j$ defined on $(0,1]$ such that $f_j>0$ a.e.~and $\lim_{x\to 0+} f_j(x)\in [0,+\infty]$, i.e.
\begin{equation}\label{eq:p_i-gf-multi-F}
\max_{1\leq k\leq n_j} \left|\frac{f_j(z_{n,j,k})}{n p_{n,j,k}} - 1\right|\to 0, \quad \text{as }\ n\to\infty;
\end{equation}
\end{enumerate}
Note that \eqref{ass:a} implies \eqref{ass:R3} and \eqref{ass:b} is the precise statement of \eqref{ass:P3}.
We have the following generalization of \prettyref{prop:H-limit}. The proof is found in \prettyref{app:proof-H-limit-J}.

\begin{proposition}\label{prop:H-limit-J}\label{PROP:H-LIMIT-J}
Assume \eqref{ass:C2}, and conditions \eqref{ass:a} and \eqref{ass:b} above. Then
\begin{equation}\label{eq:H-limit-LRU-J}
H^{\LRU}_{n} \to \sum_{j=1}^J b_j  \int_0^1  f_j(x)\Psi_j(\nu_0 f_j(x)) dx,
\end{equation}
as $n\to\infty$, where $\nu_0$ is the unique real number in $(0,\infty)$ that satisfies
\begin{equation}\label{eq:nu0-prop:H-limit-multi-F}
 \sum_{j=1}^J b_j \int_0^1  \hat \Psi_j(\nu_0 f_j(x))dx=\beta_0.
\end{equation}
\end{proposition}

\subsection{Proof of \prettyref{prop:T_n-asympt}}\label{subsec:proof-T_n-asympt}

We need the following lemmas.

\begin{lemma}\label{lem:beta}
The function 
\[
\beta(\nu):=\int_0^1 \hat \Psi(\nu f(x))dx
\]
has the following properties,
\begin{enumerate}[label=(\roman*), ref=\roman*]
\item\label{beta1} $\beta(0) = 0$, $\lim_{\nu\to\infty} \beta(\nu) = 1$;
\item\label{beta2} $\beta$ is continuous;
\item\label{beta3} $\beta$ is increasing in $\nu$;
\item\label{beta4} $\beta$ is strictly increasing at all $\nu$ such that $\beta(\nu) < 1$.
\end{enumerate}
\end{lemma}

\begin{proof}
By \prettyref{eq:F-hat}, $\hat\Psi(0) = 0$, which implies in turn implies that $\beta(0) = 0$. Since $\lim_{t\to\infty} \hat \Psi(t) = 1$, by the Bounded Convergence Theorem,
\[
\lim_{\nu\to\infty} \beta(\nu) = \int_0^1\lim_{\nu\to\infty} \hat \Psi(\nu f(x))dx = 1.
\] 
This proves \eqref{beta1}. \eqref{beta2} follows from the continuity of $\hat\Psi$ and the Bounded Convergence Theorem.

Let $\nu_1 > \nu_2$. Since $f(x) \geq 0$, it follows that $\hat \Psi(\nu_1 f(x))\geq\hat \Psi(\nu_2 f(x))$ and hence $\beta(\nu_1)\geq \beta(\nu_2)$. This proves \eqref{beta3}.

If $\beta(\nu_1) = \beta(\nu_2)$, continuity of $\hat \Psi(\nu f(x))$ implies $\hat \Psi(\nu_1 f(x))=\hat \Psi(\nu_2 f(x))$ for all $x$. If $f(x)>0$, then $ \nu_1 f(x) > \nu_2 f(x)$, and \prettyref{eq:F-hat} implies $\bar \Psi(\nu_2 f(x)) = 0$, which, by monotonicity of $\bar \Psi$, implies $\bar \Psi(t) = 0$ for all $t\geq \nu_2 f(x)$. Thus
\[
1-\hat \Psi(\nu_2 f(x)) = \int_{\nu_2 f(x)}^\infty \bar \Psi(t) dt = 0.
\]
It follows that $\hat \Psi(\nu_2 f(x)) = \ind{\{f(x)>0\}}$ for all $x\in (0,1]$. Since $\hat \Psi(\nu_2 f(x))$ is continuous in $x$ and $f$ is not identically zero, it follows that $\hat \Psi(\nu_2 f(x)) = 1$ and hence $\beta(\nu_2) = 1$. Thus $\beta(\nu_1) > \beta(\nu_2)$ if $\beta(\nu_2)<1$, which completes the proof of \eqref{beta4}.
\end{proof}

Now we prove  \prettyref{prop:T_n-asympt}.

\begin{proof} [Proof of \prettyref{prop:T_n-asympt}]
 Recall  $m_\Psi=1$ implies $G_{n,i}(x) = \Psi(\lambda_{n,i} x)$ and $\hat G_{n,i}(x) = \hat \Psi(\lambda_{n,i} x)$. We obtain from \eqref{eq:K-T} and \eqref{eq:p_i},
\begin{align*}
\frac{C_n}{n}= \frac{1}{n}K_n(T_n) =\frac{1}{n}\sum_{i=1}^n \hat G_{n,i}(T_n) = \frac{1}{n} \sum_{i=1}^n \hat \Psi(\lambda_{n,i} T_n)=\frac{1}{n}\sum_{i=1}^n \hat \Psi(p_{n,i} \Lambda_n T_n).
\end{align*}

Given any $\epsilon>0$, \eqref{eq:p_i-gf} yields that for sufficiently large $n$ and $i=1,\ldots,n$,
\begin{equation}\label{eq:p_i-bound}
(1-\epsilon)g_n f(z_{n,i}) \leq p_{n,i}  \leq (1+\epsilon) g_n f(z_{n,i}).
\end{equation}

Let $\nu_1=\limsup_{n\to\infty}  g_n\Lambda_n T_n$. Let $\{n_\ell:\ell\geq 1\}$ be the indices of a subsequence that converges to $\nu_1$, i.e.~$\nu_1 = \lim_{\ell\to\infty} g_{n_\ell} \Lambda_{n_\ell}  T_{n_\ell} $. First assume $\nu_1 < \infty$. For sufficiently large $\ell$,
\[
(1-\epsilon) (\nu_1-\epsilon) f(z_{n_\ell,i}) \leq p_{n_\ell, i}  \Lambda_{n_\ell}  T_{n_\ell}    \leq (1+\epsilon) (\nu_1+\epsilon) f(z_{n_\ell,i}).
\]
Since $\hat \Psi$ is non-decreasing, for sufficiently large $\ell$,
\begin{align*}
 \frac{1}{n_\ell}  \sum_{i=1}^{n_\ell} \hat \Psi\left((1-\epsilon)(\nu_1-\epsilon) f(z_{n_\ell,i})\right)  \leq \frac{C_{n_\ell}}{n_\ell} &= \frac{1}{n_\ell}\sum_{i=1}^{n_\ell} \hat \Psi(p_{n_\ell,i}  \Lambda_{n_\ell}  T_{n_\ell} )\\
& \leq \frac{1}{n_\ell}  \sum_{i=1}^{n_\ell} \hat \Psi\left((1+\epsilon)(\nu_1+\epsilon) f(z_{n_\ell,i})\right).
\end{align*}
Letting $\ell\to\infty$ and using the definition of the Riemann integral, we obtain 
\begin{align*}
\int_0^1 \hat \Psi((1-\epsilon)(\nu_1-\epsilon) f(x))dx  \leq \lim_{k\to\infty} \frac{C_{n_\ell}}{n_\ell} = \beta_0 \leq \int_0^1 \hat \Psi((1+\epsilon)(\nu_1+\epsilon) f(x))dx.
\end{align*}
Since $\hat \Psi$ is continuous, letting $\epsilon\to 0$ and using the Bounded Convergence Theorem, we obtain 
\[
\beta_0 = \int_0^1 \hat \Psi(\nu_1 f(x))dx=\beta(\nu_1).
\]
If $\nu_1 = +\infty$, repeating the above argument shows that
\[
\beta_0 \geq \beta(\nu)
\]
for any $\nu$, which would imply $\beta_0 \geq \lim_{\nu\to\infty} \beta(\nu) = 1$ by \prettyref{lem:beta}, a contradiction. Therefore, $\nu_1$ is finite and satisfies $\beta(\nu_1) = \beta_0$. The same argument shows that $\nu_2 = \liminf_{n\to\infty} g_n\Lambda_n T_n$ satisfies $\beta_0 = \beta(\nu_2)$. By \prettyref{lem:beta}, $\nu_1 = \nu_2 = \nu_0$, where $\nu_0\in (0,\infty)$ is the unique root of $\beta(\nu) = \beta_0$. It follows that \eqref{eq:T_n-asympt} holds.
\end{proof}

\subsection{Proof of \prettyref{prop:H-limit}}\label{subsec:proof-H-limit}

We will invoke Corollary \ref{corollary:mean-convergence} to show convergence. Assumption \eqref{ass:R2} holds by \prettyref{lem:equicontinuity}. Since $m_\Psi=1$ by \eqref{ass:R4},  \eqref{ass:C2} implies \eqref{ass:C1} for any $\beta_1 \in (\beta_0, 1)$.
 
Now we show that \eqref{ass:P1} holds. Let $A_\ell = \{x\in [0,1]:f(x) \geq 1/\ell\}$. Since $f>0$ a.e., 
$\lim_{\ell\to\infty} Leb(A_\ell)=Leb\{ x\in [0,1]: f(x)>0 \} = 1$, where $Leb$ is the Lebesgue measure on $[0,1]$. Thus there exists an $\ell_0$ such that $Leb(A_{\ell_0}^c)\leq (1-\kappa_1 \beta_0)/4$. Let $I  = [(1-\kappa_1 \beta_0)/4, 1]$ and $I_{n,i} = [\frac{i-1}{n}, \frac{i}{n}]$. Since $f$ is continuous, it is uniformly continuous on $I$ by the Heine-Cantor Theorem. For all sufficiently large $n$, $|f(x)-f(z_{n,i})|\leq \frac{1}{2\ell_0}$ if $ x\in  I_{n,i} \cap I$. Therefore, for all sufficiently large $n$,
\begin{align*}
\frac{1}{n} f(z_{n,i}) = \int_{I_{n,i}} f(z_{n,i}) dx &\geq \int_{I_{n,i}\cap I} \left[f(x) - \frac{1}{2\ell_0}\right] dx \\
&\geq \int_{I_{n,i}\cap I\cap A_{\ell_0}} \left(\frac{1}{\ell_0} - \frac{1}{2\ell_0}\right)dx = \frac{1}{2\ell_0} Leb(I_{n,i}\cap I \cap A_{\ell_0}).
\end{align*}
Summing over $i$, we obtain
\begin{align}
\bar P_n(\lceil \kappa_1 C_n\rceil) &\sim \sum_{i=\lceil \kappa_1 C_n\rceil+1}^n \frac{1}{n} f(z_{n,\sigma_i})\nonumber\\
& \geq \frac{1}{2\ell_0}\sum_{i=\lceil \kappa_1 C_n\rceil+1}^n Leb(I_{n,\sigma_i}\cap I \cap A_{\ell_0})\nonumber\\
& = \frac{1}{2\ell_0} Leb\left(\left(\bigcup_{i=\lceil \kappa_1 C_n\rceil+1}^n I_{n,\sigma_i}\right)\cap I \cap A_{\ell_0}\right)\nonumber\\
& \geq \frac{1}{2\ell_0} \left(Leb \left(\bigcup_{i=\lceil \kappa_1 C_n\rceil+1}^n I_{n,\sigma_i}\right)- Leb(I^c) - Leb(A_{\ell_0}^c)\right)\nonumber\\
& =\frac{1}{2\ell_0}\left( \sum_{i=\lceil \kappa_1 C_n\rceil+1}^n Leb(I_{n,\sigma_i}) -  Leb(I^c) -  Leb(A_{\ell_0}^c)\right)\nonumber\\
& \geq \frac{1}{2\ell_0}\left( \frac{n-\lceil \kappa_1 C_n\rceil}{n} - \frac{1}{4} (1-\kappa_1\beta_0)  - \frac{1}{4} (1-\kappa_1\beta_0)\right)\nonumber\\
&= \frac{1}{4\ell_0} (1-\kappa_1\beta_0) > 0.
\label{bound-Pn}
\end{align}
We conclude that \eqref{eq:heavy-tail} holds for $0<\gamma <\frac{1}{4\ell_0} (1-\kappa_1\beta_0)$.

Therefore, \eqref{eq:mean-convergence} holds by Corollary \ref{corollary:mean-convergence}. Then \prettyref{eq:H-limit} follows from \eqref{eq:mean-convergence} and the following lemma.

\begin{lemma}\label{lem:H-limit-TTL}
Under the assumptions of \prettyref{prop:H-limit}, 
\begin{equation}\label{eq:H-limit-TTL}
H_n^\TTL(T_n) \to \int_0^1 f(x) \Psi(\nu_0 f(x))dx,  \quad\hbox{as } n\to\infty.
\end{equation}
\end{lemma}

\begin{proof}
Recall that $G_{n,i}(x) = \Psi(\lambda_{n,i} x)$. We obtain from \eqref{eq:H-TTL}, \eqref{eq:H_i-TTL} and \eqref{eq:Y_i-conditional-pmf},
\[
H_n^{\TTL}(T_n)= \sum_{i=1}^n p_{n,i} \Psi(\lambda_{n,i} T_n)=\sum_{i=1}^n p_{n,i} \Psi(p_{n,i} \Lambda_n T_n).
\]
From \eqref{eq:T_n-asympt} and \eqref{eq:p_i-bound} the following inequalities hold, for any $\epsilon>0$ and $n$ large enough,
\[
(1-\epsilon) (\nu_0-\epsilon) f(z_{n,i}) \leq p_{n, i}  \Lambda_{n}  T_{n}    \leq (1+\epsilon) (\nu_0+\epsilon) f(z_{n,i}).
\]
The monotonicity of $\Psi$ then yields
\begin{equation}\label{eq:H-TTL-bounds}
\begin{aligned}
\frac{1-\epsilon}{n}  \sum_{i=1}^{n} f(z_{n,i}) \Psi\left((1-\epsilon)(\nu_0-\epsilon) f(z_{n,i})\right)  \leq H_n^{\TTL}(T_n)\\
 \leq \frac{1+\epsilon}{n} \sum_{i=1}^{n} f(z_{n,i})  \Psi\left((1+\epsilon)(\nu_0+\epsilon) f(z_{n,i})\right).
\end{aligned}
\end{equation}
 Letting $n\to\infty$ and using the definition of the Riemann integral, we obtain
\begin{equation}\label{eq:H-liminf}
\liminf_{n\to\infty} H_n^{\TTL}(T_n) \geq (1-\epsilon) \int_0^1 f(x)\Psi((1-\epsilon)(\nu_0-\epsilon) f(x) dx ,
\end{equation}
and
\begin{equation}\label{eq:H-limsup}
\limsup_{n\to\infty} H_n^{\TTL}(T_n) \leq (1+\epsilon) \int_0^1 f(x)\Psi((1+\epsilon)(\nu_0+\epsilon) f(x)) dx.
\end{equation}
The existence of the integrals comes from the fact that $0\leq \Psi \leq 1$ and the integrability of $f$ over $[0,1]$, which follows from the first inequality in \eqref{eq:p_i-bound} by the following,
\[
1 = \sum_{i=1}^n p_{n,i} \geq (1-\epsilon) \frac{1}{n}\sum_{i=1}^n f(z_{n,i})\to (1-\epsilon)  \int_0^1 f(x)dx.
\]
Since $\Psi$ is continuous and $\int_0^1 f(x)dx<\infty$,  letting $\epsilon\to 0$ in \prettyref{eq:H-liminf} and \prettyref{eq:H-limsup}  yields \prettyref{eq:H-limit-TTL} by the Dominated Convergence Theorem.
\end{proof}


\section{Conclusions}\label{sec:concl}

In this paper, we developed an approximation for the aggregate and individual content hit probability of an LRU cache based on a transformation to the TTL cache for the case that content requests are described by independent stationary and ergodic processes.  This approximation extends one first proposed and studied by Fagin \cite{fagin1977asymptotic} for the independent reference model and provides the theoretical basis for approximations introduced in \cite{Leonardi16} for content requests described by independent renewal processes.  We showed that the approximations become exact in the limit as the cache size and the number of contents go to infinity.   Last, we established the rate of convergence for the approximation as number of contents increases.

Future directions include investigation for tighter bounds on the convergence rate and extension of these results to other cache policies such as FIFO and random and to networks of caches perhaps using ideas from \cite{rosensweig,TTL-ValueTools-14,berger14}. In addition, it is desirable to relax independence between different content request streams.


\bibliographystyle{ACM-Reference-Format}
\bibliography{TEX/bib} 

\appendix


\section{Equicontinuity}

\begin{lemma}\label{lem:equicontinuity}
A finite family of continuous cdfs is equicontinuous, so  that  \eqref{ass:R4}$\implies$\eqref{ass:R3}$\implies$\eqref{ass:R2}.
\end{lemma}

\begin{proof}
Let the family of cdfs be $\F = \{F_1,\dots,F_J\}$. Fix $\epsilon$. There exists a $L_j\in (0,\infty)$ such that 
\begin{equation}\label{eq:equicontinuity-L}
F_j(-L_j)<\epsilon \quad \text{ and } \quad 1- F_j(L_j)<\epsilon.
\end{equation}
Let $L = \max_{1\leq j\leq J} L_j \in (0,\infty)$. Being continuous, $F_j$ is uniformly continuous on $[-2L, 2L]$ by the Heine-Cantor Theorem. Thus there exists a $\delta_j\in (0, L)$ such that 
\begin{equation}\label{eq:equicontinuity-delta}
|F_j(x_1) - F_j(x_2)|< \epsilon,
\end{equation}
for $x_1,x_2\in [-2L, 2L]$ such that $|x_1 - x_2|<\delta_j$.

Let $\delta = \min_{1\leq j\leq J} \delta_j\in (0,L)$. Consider any $x_1>x_2$ with $|x_1-x_2|<\delta$. There are three cases.
\begin{enumerate}[label=(\roman*)]
\item If $x_1,x_2\in [-2L, 2L]$, then \eqref{eq:equicontinuity-delta} holds for all $j$.
\item If $x_1>2L$, then $x_2 > L$, since $|x_1-x_2|<\delta<L$. Thus $|F_j(x_1)-F_j(x_2)| = F_j(x_1)-F_j(x_2)\leq 1-F_j(L_j)<\epsilon$ by \eqref{eq:equicontinuity-L}, and this holds for all $j$.
\item If $x_2<-2L$, then $x_1 < -L$, since $|x_1-x_2|<\delta<L$. Thus $|F_j(x_1)-F_j(x_2)| = \leq F_j(-L_j)<\epsilon$ by \eqref{eq:equicontinuity-L}, and this holds for all $j$.
\end{enumerate}
Therefore, $\F$ is equicontinuous.
\end{proof}

\section{Proof of \prettyref{PROP:H-LIMIT-J}}\label{app:proof-H-limit-J}

The proof parallels those of Proposition \ref{prop:T_n-asympt} and \ref{prop:H-limit} except for the last step.
The following lemma generalizes \prettyref{lem:beta}.

\begin{lemma}\label{lem:beta-J}
The function 
\[
\beta_J(\nu):=\sum_{j=1}^J b_j \int_0^1 \hat \Psi_j(\nu f_j(x))dx
\]
has the following properties,
\begin{enumerate}[label=(\roman*), ref=\roman*]
\item\label{beta_J1} $\beta_J(0) = 0$, $\lim_{\nu\to\infty}\beta_J(\nu) = 1$;
\item\label{beta_J2} $\beta_J$ is continuous;
\item\label{beta_J3} $\beta_J$ is increasing in $\nu$;
\item\label{beta_J4} $\beta_J$ is strictly increasing at all $\nu$ such that $\beta_J(\nu) < 1$.
\end{enumerate}
\end{lemma}

\begin{proof}
By \prettyref{eq:Psi-hat}, $\hat\Psi_j(0) = 0$, which implies $\beta_J(0) = 0$. Since $\lim_{x\to\infty} \hat \Psi_j(x) = 1$, by the Bounded Convergence Theorem,
\[
\lim_{\nu\to\infty} \beta_J(\nu) =\sum_{j=1}^J b_j \int_0^1\lim_{\nu \to\infty} \hat \Psi_j(\nu f_j(x))dx = 1.
\] 
This proves \eqref{beta_J1}. \eqref{beta_J2} follows from the continuity of $\hat\Psi_j$ and the Bounded Convergence Theorem.

Let $\nu_1 > \nu_2$. Since $f_j(x) \geq 0$, it follows that $\hat \Psi_j(\nu_1 f_j(x))\geq\hat \Psi_j(\nu_2 f_j(x))$ and hence $\beta_J(\nu_1)\geq \beta_J(\nu_2)$. This proves \eqref{beta_J3}.

If $\beta_J(\nu_1) = \beta_J(\nu_2)$, then continuity of $\hat \Psi_j(\nu f_j(x))$ implies that $\hat \Psi_j(\nu_1 f_j(x))=\hat \Psi_j(\nu_2 f_j(x))$ for all $x$. If $f_j(x)>0$, then $ \nu_1 f_j(x) > \nu_2 f_j(x)$, and \prettyref{eq:F-hat} implies $\bar \Psi_j(\nu_2 f_j(x)) = 0$, which, by monotonicity of $\bar \Psi_j$, implies $\bar \Psi_j(t) = 0$ for all $t\geq \nu_2 f_j(x)$. Thus
\[
1-\hat \Psi_j(\nu_2 f_j(x)) = \int_{\nu_2 f_j(x)}^\infty \bar \Psi_j(t) dt = 0.
\]
It follows that $\hat \Psi_j(\nu_2 f_j(x)) = \ind{\{f_j(x)>0\}}$ for all $x\in (0,1]$. Since $\hat \Psi_j(\nu_2 f_j(x))$ is continuous in $x$ and $f_j$ is not identically zero, it follows that $\hat \Psi_j(\nu_2 f_j(x)) = 1$. Since this is true for all $j$, it follows that $\beta_J(\nu_2) = 1$. Thus $\beta_J(\nu_1) > \beta_J(\nu_2)$ if $\beta_J(\nu_2)<1$, which completes the proof of \eqref{beta_J4}.
\end{proof}

The following proposition generalizes \prettyref{prop:T_n-asympt}.

\begin{proposition}
\label{prop:T_n-asympt-generalize}
Under the assumptions in \prettyref{prop:H-limit-J} but with the condition \eqref{ass:b} that $p_{n,j,k} \sim n^{-1} f_j(z_{n,j,k})$ generalized to $p_{n,j,k} \sim g_n f_j(z_{n,j,k})$, we have
\begin{equation}\label{eq:T_n-asympt-J}
T_n \sim \frac{\nu_0}{g_n\Lambda_n},
\end{equation}
where $\nu_0$ satisfies \eqref{eq:nu0-prop:H-limit-multi-F}.
\end{proposition}

\begin{proof} 
 Recall   $G_{n,j,k}(x) = \Psi_j(\lambda_{n,j,k} x)$ implies $\hat G_{n,j,k}(x) = \hat \Psi_j(\lambda_{n,j,k} x)$. We obtain from \eqref{eq:K-T} and \eqref{eq:p_i},
\begin{align*}
\frac{C_n}{n} = \frac{1}{n}K_n(T_n) =\frac{1}{n}\sum_{j=1}^{J} \sum_{k=1}^{b_j n} \hat G_{n,j,k}(T_n) = \frac{1}{n} \sum_{j=1}^{J} \sum_{k=1}^{b_j n} \hat \Psi_j(\lambda_{n,j,k} T_n)=\frac{1}{n}\sum_{j=1}^{J} \sum_{k=1}^{b_j n} \hat \Psi_j(p_{n,j,k} \Lambda_n T_n).
\end{align*}

Given any $\epsilon>0$, \eqref{eq:p_i-gf-multi-F} yields that for sufficiently large $n$ and $j=1,\dots,J$, $k=1,\ldots,b_j n$,
\begin{equation}\label{eq:p_i-bound-J}
(1-\epsilon)g_n f_j(z_{n,j,k}) \leq p_{n,j,k}  \leq (1+\epsilon) g_n f_j(z_{n,j,k}).
\end{equation}

Let $\nu_1=\limsup_{n\to\infty}  g_n\Lambda_n T_n$.
Let $\{n_\ell:\ell\geq 1\}$ be the indices of a subsequence that converges to $\nu_1$, i.e.~$\nu_1 = \lim_{\ell\to\infty} g_{n_\ell} \Lambda_{n_\ell}  T_{n_\ell}$. First assume $\nu_1 < \infty$ for all $j$. For sufficiently large $\ell$,
\[
(1-\epsilon) (\nu_1-\epsilon) f_j(z_{n_\ell,i}) \leq p_{n_\ell, j,k}  \Lambda_{n_\ell}  T_{n_\ell}    \leq (1+\epsilon) (\nu_1+\epsilon) f_j(z_{n_\ell,i}).
\]
Since $\hat \Psi_j$ is non-decreasing, for sufficiently large $\ell$,
\begin{align*}
\frac{1}{n_\ell}  \sum_{j=1}^J\sum_{k=1}^{b_j n_\ell} \hat \Psi_j\left((1-\epsilon)(\nu_1-\epsilon) f_j(z_{n_\ell,i})\right)  \leq \frac{C_{n_\ell}}{n_\ell} &= \frac{1}{n_\ell} \sum_{j=1}^J\sum_{k=1}^{b_j n_\ell}\hat \Psi_j(p_{n_\ell,i}  \Lambda_{n_\ell}  T_{n_\ell} )\\
& \leq \frac{1}{n_\ell} \sum_{j=1}^J\sum_{k=1}^{b_j n_\ell} \hat \Psi_j\left((1+\epsilon)(\nu_1+\epsilon) f_j(z_{n_\ell,i})\right).
\end{align*}
Letting $\ell\to\infty$ and using the definition of the Riemann integral, we obtain 
\begin{align*}
\sum_{j=1}^J b_j \int_0^1 \hat \Psi_j((1-\epsilon)(\nu_1-\epsilon) f_j(x))dx \leq \lim_{\ell\to\infty} \frac{C_{n_\ell}}{n_\ell} = \beta_0 \leq \sum_{j=1}^J b_j  \int_0^1 \hat \Psi_j((1+\epsilon)(\nu_1+\epsilon) f_j(x))dx.
\end{align*}
Since $\hat \Psi_j$ is continuous, letting $\epsilon\to 0$ and using the Bounded Convergence Theorem, we obtain 
\[
\beta_0 = \sum_{j=1}^J b_j \int_0^1 \hat \Psi_j(\nu_1 f_j(x))dx=\beta_J(\nu_1).
\]
If $\nu_1 = +\infty$,  repeating the above argument shows that
\[
\beta_0\geq \beta_J(\nu)
\]
for any $\nu$, which would imply $\beta_0 \geq \lim_{\nu\to\infty} \beta_J(\nu) = 1$ by \prettyref{lem:beta-J}, a contradiction. Therefore, $\nu_1$ is finite and satisfies $\beta_J(\nu_1) = \beta_0$.
The same argument shows that $\nu_2 = \liminf_{n\to\infty} g_n\Lambda_n T_n$ satisfies $\beta_0 = \beta_J(\nu_2)$. By \prettyref{lem:beta-J}, $\nu_1 = \nu_2 = \nu_0$, where $\nu_0\in (0,\infty)$ is the unique root of $\beta_J(\nu) = \beta_0$. It follows that \eqref{eq:T_n-asympt-J} holds.
\end{proof}


The following lemma generalizes \prettyref{lem:H-limit-TTL}.

\begin{lemma}\label{lem:H-limit-TTL-J}
Under the assumptions of \prettyref{prop:H-limit-J}, 
\begin{equation}\label{eq:H-limit-TTL-J}
H_n^\TTL(T_n) \to \sum_{j=1}^J b_j \int_0^1 f_j(x) \Psi_j(\nu_0 f_j(x)) dx.
\end{equation}
\end{lemma}

\begin{proof}
Recall that $G_{n,j,k}(x) = \Psi_j(\lambda_{n,j,k} x)$. We obtain from \eqref{eq:H-TTL}, \eqref{eq:H_i-TTL} and \eqref{eq:Y_i-conditional-pmf},
\begin{align*}
H_n^{\TTL}(T_n)=\sum_{j=1}^J \sum_{i=1}^{b_j n} p_{n,j,k} \Psi_j(\lambda_{n,j,k} T_n)=\sum_{j=1}^J \sum_{i=1}^{b_j n} p_{n,j,k} \Psi_j(p_{n,j,k} \Lambda_n T_n).
\end{align*}
Given any $\epsilon>0$, for all sufficiently large $n$, \eqref{eq:p_i-bound-J} and the following hold,
\begin{equation}
\label{inq:fj}
(1-\epsilon) (\nu_0-\epsilon) f_j(z_{n,j,k}) \leq p_{n,j,k}  \Lambda_{n}  T_{n}    \leq (1+\epsilon) (\nu_0+\epsilon) f_j(z_{n,j,k}).
\end{equation}
The monotonicity of $\Psi_j$ then yields
\begin{equation*}
\begin{aligned}
\frac{1-\epsilon}{n}  \sum_{j=1}^J \sum_{k=1}^{b_j n} f_j(z_{n,j,k}) \Psi_j\left((1-\epsilon)(\nu_0-\epsilon) f_j(z_{n,j,k})\right)  \leq H_n^{\TTL}(T_n)\\
 \leq \frac{1+\epsilon}{n} \sum_{j=1}^J \sum_{k=1}^{b_j n}f_j(z_{n,j,k})  \Psi_j\left((1+\epsilon)(\nu_0+\epsilon) f_j(z_{n,j,k})\right).
\end{aligned}
\end{equation*}
Letting $n\to\infty$ and using the definition of the Riemann integral, we find
\begin{equation}\label{eq:H-liminf-J}
\liminf_{n\to\infty} H_n^{\TTL}(T_n) \\
\geq  (1-\epsilon)\sum_{j=1}^J b_j \int_0^1 f_j(x)\Psi_j((1-\epsilon)(\nu_0-\epsilon) f_j(x) dx ,
\end{equation}
and
\begin{equation}\label{eq:H-limsup-J}
\limsup_{n\to\infty} H_n^{\TTL}(T_n)\\
 \leq  (1+\epsilon) \sum_{j=1}^J b_j \int_0^1 f_j(x)\Psi_j((1+\epsilon)(\nu_0+\epsilon) f_j(x)) dx.
\end{equation}
The existence of the integrals comes from the fact that $0\leq \Psi_j \leq 1$ and the integrability of $f_j$ over $[0,1]$, which follows from the first inequality in \eqref{eq:p_i-bound-J} by the following,
\begin{align}
1 = \sum_{j=1}^J \sum_{k=1}^{b_j n} p_{n,j,k} \geq (1-\epsilon) \frac{1}{n}\sum_{j=1}^J \sum_{k=1}^{b_j n} f_j(z_{n,j,k})\to (1-\epsilon)  \sum_{j=1}^J b_j \int_0^1 f_j(x)dx\label{eq:f_j-normalization}.
\end{align}
Since $\Psi_j$ is continuous and $\int_0^1 f_j(x)dx<\infty$,  letting $\epsilon\to 0$ in \prettyref{eq:H-liminf-J} and \prettyref{eq:H-limsup-J}  yields \prettyref{eq:H-limit-TTL-J} by the Dominated Convergence Theorem.
\end{proof}

\begin{proof}[Proof of \prettyref{prop:H-limit-J}]
Thanks to Lemma \ref{lem:H-limit-TTL-J} and the value of $\nu_0$ given in Proposition \ref{prop:T_n-asympt-generalize} that satisfies (\ref{eq:nu0-prop:H-limit-multi-F}), we only need to show the convergence of $H_n^\LRU$ to $H_n^\TTL(T_n)$ as $n\to\infty$.
 For that, we invoke  Corollary \ref{corollary:mean-convergence}.
We use Remark \ref{remark-prop.4-4} and show that \eqref{eq:mu-lower} holds under the conditions of \prettyref{prop:H-limit-J}. 
Repeating the proof of \eqref{eq:H-limit-TTL-J}, we obtain
\begin{equation}
H_n^\TTL((1+x)T_n) \to \sum_{j=1}^J b_j \int_0^1 f_j(y) \Psi_j((1+x)\nu_0 f_j(y)) dy,
\label{HTTLxTn}
\end{equation}
as $n\to\infty$.
Fix $\epsilon>0$. Summing (\ref{inq:fj}) over $j$ and $k$ and letting $g_n=1/n$ yields, for $n$ large enough,
\[
(1-\epsilon)\frac{1}{n}\sum_{j=1}^J\sum_{k=1}^{b_jn}  f(z_{n,i,k})\leq \sum_{j=1}^J\sum_{k=1}^{b_jn} p_{n,j,k}=1\leq (1+\epsilon)\frac{1}{n}\sum_{j=1}^J \sum_{k=1}^{b_jn}f(z_{n,i,k}).
\]
Letting $n\to\infty$ we obtain, by the definition of the Riemann integral, 
\[
(1-\epsilon) \sum_{j=1}^J b_j \int_0^1 f_j(y)  dy\leq 1\leq (1+\epsilon)  \sum_{j=1}^J b_j \int_0^1 f_j(y)  dy,
\]
from which we conclude that
\begin{equation}
\label{normalize-fj}
1=\sum_{j=1}^J b_j \int_0^1 f_j(y)  dy.
\end{equation}
Since $\mu_n(T) = 1-H_n^\TTL(T)$, subtracting (\ref{HTTLxTn}) from (\ref{normalize-fj})  yields
\[
\mu_n((1+x)T_n)  \to \sum_{j=1}^J b_j \int_0^1 f_j(y) \bar \Psi_j((1+x)\nu_0 f_j(y)) dy:=\mu(x).
\]
Note that $\mu$ is continuous by the continuity of $f_j$, $\Psi_j$ and the Dominated Convergence Theorem. If $\mu(0)>0$, then there exists $x_0>0$ such that $\mu(x_0) 
\geq \mu(0)/2>0$. Thus for sufficiently large $n$ and $T\leq (1+x_0) T_n$,
\[
\mu_n(T) \geq \mu_n((1+x_0)T_n) \geq \mu(x_0)/2 \geq \mu(0)/4 > 0.
\]
Since $C_n \leq \Lambda_n T_n$ by \eqref{eq:T_n-lower} with $n_2=0$, the above inequality yields \eqref{eq:mu-lower} with $\phi = \mu(0)/4$, provided that $\mu(0)>0$.

Now we show that $\mu(0) > 0$. Suppose $\mu(0) = 0$. Then 
\[
\int_0^1 f_j(y) \bar \Psi_j(\nu_0 f_j(y)) dy=0
\]
for each $j$. Since $f_j(y)>0$ a.e. on $(0,1]$ and $\bar\Psi_j(\cdot)\geq 0$ for each $j$, it follows that $\bar\Psi_j(\nu_0 f_j(y))=0$ a.e.~on $(0,1]$ for each $j$.
Hence, by (\ref{eq:F-hat}) with $m_{\Psi_j}=1$ for each $j$,
\[
\beta_J(\nu_0)=\sum_{j=1}^J b_j  \int_0^1 \hat \Psi_j(\nu_0 f_j(x))dx= \sum_{j=1}^J b_j  \int_0^1  \int_0^{ \nu_0 f_j(x)} \bar \Psi(y) dy dx=0,
\]
which contradicts the result obtained in Proposition \ref{prop:T_n-asympt-generalize} that $\beta_J(\nu_0)=\beta_0\in (0,1)$ .
Therefore, $\mu(0)>0$, which completes the proof.
 \end{proof}

\begin{remark}
\label{alternative-proof}
A proof similar to that of Proposition \ref{prop:H-limit} can be done if we let $\Psi = \max_{1\leq j\leq J} \Psi_j$ and restrict the range of $\beta_0$ to $\beta_0< m_\Psi=
\int_0^\infty \min_{1\leq j\leq J} \bar \Psi_j(x)dx$ (see Section \ref{subsec:request}), a quantity that is in general strictly less than one.  
This restriction on $\beta_0$ is a consequence of condition \eqref{ass:C1}.
It is also worth noting that, since  Proposition \ref{prop:H-limit-J} reduces to Proposition \ref{prop:H-limit} when $J=1$, the proof of Proposition \ref{prop:H-limit-J}  provides an alternative proof of
 Proposition \ref{prop:H-limit}.
\end{remark}

\end{document}